%% file: degreebounds.tex
\documentclass[runningheads,12pt]{article}

\input{header.tex}
\usepackage{xparse}
\DeclareDocumentCommand{\D}{O{t} O{1} O{S} }{d_{#2}^{#1}(#3)}

\title{The Fine Structure of\\Preferential Attachment Graphs I:\\
Somewhere-Denseness}
\author{Jan Dreier, Philipp Kuinke, Peter Rossmanith\\
        \small Theoretical Computer Science, RWTH Aachen University}

\begin{document}

\maketitle

\begin{abstract}
Preferential attachment graphs are random graphs designed to mimic
properties of typical real word networks. They
are constructed by a random process that iteratively adds vertices and 
attaches them preferentially to vertices that already have high degree.
We use improved concentration bounds for vertex degrees to show
that preferential attachment graphs contain 
asymptotically almost surely (\aas) a one-subdivided
clique of size at least $(\log n)^{1/4}$.
Therefore, preferential attachment graphs
are \aas somewhere-dense.
This implies that algorithmic
techniques developed for sparse graphs are not directly applicable to them.
The concentration bounds state: 
Assuming that the exact degree $d$ of a fixed vertex (or set of vertices) at some early time $t$ of the random process is known,
the probability distribution of $d$ is sharply concentrated
as the random process evolves if and only if $d$ is large at time $t$.
\end{abstract}

\section{Introduction}
\label{sec:intro}

Recently, there were more and more efforts to model real world networks using
random graph models. One property that was observed for some  real world
networks is
that they might follow a power law degree distributions and are scale-free,
this has been studied in detail~\cite{broido2018scale,clauset2009power}.
One candidate to model these properties are the
\BaAlGrs or \emph{preferential attachment graphs}, with a degree distribution that mimics the heavy-tailed
distribution observed in many real-world
networks~\cite{barabasi1999emergence}.

These preferential attachment graphs are created by a random process that
iteratively adds new vertices
and randomly connects them to already existing ones.
The model has a free parameter $m$, that indicates the number of edges (or self-loops) attached to each newly created
vertex. A random graph is started with one vertex that has exactly $m$
self-loops.
Afterwards, more vertices are added iteratively. Every time a new vertex is
added, $m$
random edges between this vertex and existing vertices are added.
The probability that an edge
from the new vertex to another vertex $v$ is added is proportional to the
current degree of $v$ (see Section~\ref{sec:prelim} for a more rigorous
definition). This process creates a certain imbalance: The degree of
low-degree vertices is unlikely to increase and the degree of high-degree
vertices is likely to increase even further during the process.
This effect is called ``the rich get
richer''-effect due to the tendency of high degree vertices to accumulate even
more edges in the future.
Vertices that have been added
early in the random process have a higher expected degree
than late vertices.

The preferential attachment model
is particularly interesting from the point of mathematical analysis
because of its simple formulation and interesting characteristics.
The model has been widely studied in the
literature~\cite{cohen,kamrul,klemm}. Other random
graph models that were created to
simulate real world graphs include the
Chung-Lu-~\cite{chung2002average,chung2002connected}, the
Configuration-~\cite{molloy1995critical} and the
Kleinberg-model~\cite{kleinberg2000small,kleinberg2000navigation}.

Another often observed (and algorithmically exploitable) property of real
networks is that they are
sparse~\cite{StrucSpars}.
In a social network there may be some hubs connected to a lot of people,
but the vast majority of members only has relatively few neighbors,
especially compared to the set of all potential neighbors.
Sparsity is a concept that has been deeply studied and has lead to
interesting algorithmic applications.
Many graph problems that are hard for general graphs become easier on sparse
graph classes~\cite{bodlaender2016meta,courcelle1990graph,eickmeyer_et_al,grohe2017deciding,grohe2017first}.
Before we can talk about algorithms, we have to discuss
formally what it means for a graph to be sparse.
One cannot simply say that a graph is sparse if it contains few edges.
The one-subdivision of any graph has roughly two edges per vertex but has the
same dense structure as the original graph and should therefore
not be considered sparse.
There are a number of measures for graph
sparsity and one of the best known
and most general, introduced by \Nesetril and Ossona
de Mendez, is \emph{nowhere-denseness}~\cite{nevsetvril2012sparsity}.
This concept generalizes many different sparse graph properties
such as bounded degree, planarity, bounded treewidth, bounded genus or bounded
expansion.
Informally speaking, a graph class is nowhere-dense if one needs a subgraph
with a large radius to construct a large clique-minor, where large
means growing with $n$. A graph class that is not nowhere-dense is
said to be \emph{somewhere-dense}.
For a rigorous definition of these two concepts see
Section~\ref{sec:prelim:sparsity}.
While some graph properties like planarity are properties of a single graph,
it should be noted that nowhere- and somewhere-denseness are not.
They are properties of a graph class, just like having bounded degree.
Let $\cal C$ be a nowhere-dense graph class.
It has been shown by Grohe, Kreutzer and Siebertz that
for any fixed first-order formula
it is possible to decide whether it is satisfied in a graph from $\cal C$ in
linear time (in the size of the graph)~\cite{grohe2017deciding}.
This leads to efficient algorithms on nowhere-dense
graph classes for every problem that can be
expressed as a short first-order formula, e.g.,
the parameterized clique or dominating set problem.

In this work we aim to transfer the deep structural and algorithmic results
for sparse graph classes to random graph models.
The concept of nowhere-denseness is defined for deterministic graph classes
which means we have to find a definition that is suitable for random graph
models. In this work (in line with earlier work~\cite{StrucSpars,nevsetvril2012sparsity}), we ask
for a given random graph model whether there exists a nowhere-dense graph
class $\cal C$ such that with probability $1$ an instance with $n$ vertices
from this model belongs to $\cal C$ as $n$ approaches infinity.
If this is the case then
the random graph model is \emph{asymptotically almost surely}
(\aas) nowhere-dense. Similarly, if for all nowhere-dense graph classes
$\cal C$ an instance from this model belongs to
$\cal C$ with probability $0$ as $n$ approaches infinity, the random graph model is
\aas somewhere-dense. Again, for more details see
Section~\ref{sec:prelim:sparsity}.

We believe \aas somewhere- and \aas nowhere-denseness are important
properties of random graph models,
which reveal a lot of information about asymptotic structure
and can help in the design of efficient algorithms.
Efficient algorithms for random graph models 
which mimic real world networks may lead to
efficient algorithms for real world networks as well.

If a random graph model is \aas nowhere-dense it is not \aas
somewhere-dense, and vice versa.
However, it is possible that a random graph model is neither
\aas somewhere- nor \aas nowhere-dense, which can not happen for deterministic
graph classes.
To become more familiar with these concepts consider
the following three simple random graph models:
A graph with $n$ vertices is
\begin{enumerate}
    \item with probability $1/n$ complete and with probability $1-1/n$ empty,
    \item with probability $1/2$ complete and with probability $1/2$ empty,
    \item with probability $1-1/n$ complete and with probability $1/n$ empty.
\end{enumerate}
If we fix a clique-size $k$, the probability that a graph from these models
contains a $k$-clique as a subgraph
converges to $0$, $1/2$ and $1$, respectively
as $n$ approaches infinity.
Therefore, the first model is \aas nowhere-dense, the third
is \aas somewhere-dense and the second is neither
\aas somewhere- nor \aas nowhere-dense.

So far it is only known
that preferential attachment graphs are \emph{not} \aas
nowhere-dense~\cite{StrucSpars}.
In this work we show that they are \aas somewhere-dense, which
has stronger implications.
If a random graph model is \aas nowhere-dense the 
well known model-checking algorithm for sparse graphs~\cite{grohe2017deciding}
has with probability $1-\varepsilon$ an efficient running time, where
$\varepsilon$ converges to zero with increasing graph size.
If a random graph model is \aas somewhere-dense then techniques 
developed for sparse graphs are not directly applicable.
An efficient model-checking algorithm for this graph model (if it exists)
has to exploit additional information about this graph model.
If, however, a random graph model is only known to be not \aas nowhere-dense
the picture is less clear.
In this case we know that the probability of belonging to a fixed sparse graph
class does not converge to zero.
But it could be that for every $\varepsilon > 0$ with probability $1-\varepsilon$ 
a graph from this model belongs to some sparse graph class that depends on $\varepsilon$.
Then with probability $1-\varepsilon$ the model-checking
algorithm for nowhere-dense graphs runs in linear time.
But this technique does not work for every random graph model that
is not \aas nowhere-dense.
Note that both aforementioned algorithms do not necessarily have an efficient
expected running time
and in the long run we would be interested in algorithms that have
expected linear running time.

\paragraph{Somewhere-dense.}

While many random graph models have already been classified (the
Chung-Lu- and the configuration-model are \aas nowhere-dense and the
Kleinberg-Model is \aas somewhere-dense~\cite{StrucSpars}), it remained an
open question whether preferential attachment graphs are \aas somewhere-dense.
It is known
that they are \emph{not} \aas nowhere-dense, i.e., there is a non vanishing
probability that there is a clique-minor of arbitrary size~\cite{StrucSpars}.
Our main result (Section~\ref{sec:somewhere-dense})
is that graphs constructed by the preferential attachment
model are somewhere-dense in the limit, thus answering this question.
We do so by showing that the probability that a preferential attachment
graph of size $n$ contains a one-subdivided clique of size $\log(n)^{1/4}$ as a
subgraph approaches one as $n$ approaches infinity.
This means there exists a somewhere-dense graph class such that
the probability that a preferential attachment graph $G_m^n$ of size $n$
and parameter $m$
belongs to this graph class approaches one as $n$ approaches infinity.
\begin{customthm}{\ref{thm:somewhere_dense}}
    Let $m \ge 2$.
    $G_m^n$ contains
    with a probability of at least $1 - \log(n)^{O(1)} e^{-c \log(n)^{1/4}}$
    a one-subdivided clique of size
    $\lfloor \log(n)^{1/4} \rfloor$
    for some positive constant $c$.
\end{customthm}
This asymptotic property might help to indicate which real-world
networks
can or cannot be modeled by a preferential attachment paradigm.
A scale-free real-world network that follows the preferential
attachment paradigm should also contain one-subdivided cliques of
increasing size.
Our results also imply that efficient model-checking algorithms
for preferential attachment graphs need to be based on a deeper analysis
of their structure.
In the upcoming follow-up work we will show that the size of subdivided
cliques in
preferential attachment graphs does not exceed polylogarithmic size.
This means that the density is unbounded, but only grows slowly.
We will use this and further observations to construct efficient algorithms for
preferential attachment graphs, which may lead
to efficient algorithms for real networks in the future.

\begin{figure}
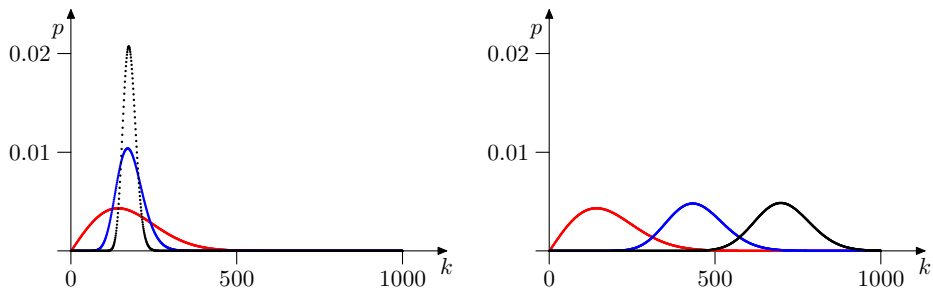

\centerline{\includegraphics[width=0.4\textwidth]{code/plot.2}\quad
\includegraphics[width=0.4\textwidth]{code/plot.1}}
\caption{
\emph{Left:} Probability distribution of $\D[10000][1][v_1]$.
Unconditional (red), and under the conditions $\D[100][1][v_1]=18$ (blue),
and $\D[1000][1][v_1]=56$ (black).
\emph{Right:} Probability distribution of $\D[10000][1][v_1]$,
$\D[10000][1][\{v_1,\dots,v_{20}\}]$ and
$\D[10000][1][\{v_1,\dots,v_{50}\}]$ (in red, blue and black respectively),
i.e., the sum of the degrees of the first 1, 20 and 50 vertices after 10000
steps. It looks similar to a shifted normal distribution with the same
variance.
}
\label{fig:vertex_conditional}
\end{figure}

\paragraph{Tail Bounds.}
We also present a detailed analysis of the probability
distribution of the degree of individual vertices during the preferential
attachment process, including exponentially decreasing tail bounds.
The results are crucial for showing that preferential attachment
graphs are \aas somewhere-dense but may also be interesting for other applications.

The preferential attachment process depends on a parameter $m \in \N$ which
states how many edges are added per vertex.
Let $t \in \N$ and $S$ be a set of vertices in a preferential attachment graph
of size $t$.
For $n \ge t$ we define $\D[n][m][S]$ to be the sum over all degrees of vertices in
$S$ in an preferential attachment graph of size $n$ with parameter $m$.
The earlier a vertex had been added in the random
process, the more time it had to accumulate neighbors.
It can be shown that the expected degree of the $i$th vertex in a preferential
attachment graph of size $n$, i.e., $\E[\D[n][m][v_i]]$ is approximately
$m\sqrt{n/i}$~\cite{hofstad1}.
But the preferential attachment process is
too unstable and chaotic to guarantee that the degree of a vertex
is closely centered around its expected degree.
In Figure~\ref{fig:vertex_conditional} on the left
the exact probability distribution of the first vertex in a preferential
attachment graph of size $10000$ is plotted in red.
The probability that the $i$th vertex has only degree one at time $n$ is at
least $1/n$, i.e., quite high.
See Lemma~\ref{lem:first_vertex_bound_special} 
for stronger bounds for the first vertex.

If we, however, assume that the degree of a vertex already
is relatively high without changing its expectation, things start to change:
the exact probability distribution of the first vertex under the condition
that it had roughly its expected value after $100$ ($1000$) steps is
plotted in blue (black).
We observe that, since the expected value did not change much,
the distribution became much more concentrated
than in the unconditional distribution.
On the right-hand side of Figure~\ref{fig:vertex_conditional} the distribution for
the sum of the degrees of a vertex set of size $1$, $20$ and $50$ is shown
after $10000$ steps. The absolute concentration of the probability
distribution of the total degree does not change if we increase the size of
the set, but since the curve is moved to the right the relative concentration
is increasing.
We show that if we assume $\D[t][m][S]$ to be high,
then the probability is also high, that for all $n>t$ the degree $\D[n][m][S]$ 
is no more than a constant factor off
from the expected degree $\E[\D[n][m][S] \mid \D[t][m][S]$.
Also, the probability that $\D[n][m][S]$ differs by more than a constant
factor decreases exponentially for large $\D[t][m][S]$.
This is formalized by the following theorem.
Note that the constants in this theorem and its proof
are chosen to make the calculations easier and can most likely be improved.

\begin{customthm}{\ref{thm:degree-bounds}}
    Let $0 < \varepsilon \le 1/40$, $t,m \in \N$,
    $t > \frac{1}{\varepsilon^6}$ and
    $S \subseteq \{v_1, \dots, v_t\}$.
    Then
    \begin{multline*}
        \P\Bigl[  (1-\varepsilon) \sqrt{\frac{n}{t}} \D[t][m] < \D[n][m]
            < (1+\varepsilon) \sqrt{\frac{n}{t}} \D[t][m] \text{ for all $n \ge t$ }
        \Bigm| \D[t][m] \Bigr] \\ \ge
        1 -  \ln(15t) e^{\varepsilon^{-O(1)}\D[t][m]}.
    \end{multline*}
\end{customthm}

Informally speaking, we show that, while some members of a preferential
attachment process behave chaotically,
central hubs keep growing at a very predictable rate.
The more important a member becomes during the preferential attachment process
the higher and more predictable its growth rate becomes.
One can say that
not only ``the rich get richer'' but also ``the rich are predictable.''

We use this result in our proof that preferential attachment graphs are \aas
somewhere-dense.
We choose a set of so called principal vertices with high degree,
Theorem~\ref{thm:degree-bounds} then guarantees that the degrees of the
principal vertices are at all times $n$ of the random process
high, i.e., approximately $\sqrt{n}$.
We then show that because their degree is high the principal vertices
will become pairwise connected and therefore span a subdivided
clique, which completes the proof.

\Bollobas and Spencer have provided tail bounds on the number of vertices
with degree $d$, where $d$ is small~\cite{Bollobas:2001}.
These tail bounds were obtained via martingales and the Azuma--Hoeffding
inequality.
This technique, however, breaks down if one tries to generalize the result
to higher $d$, i.e., $d$ of order $\sqrt{n}$.
Therefore, we present a novel approach to prove tail bounds.
We use Chernoff bounds to obtain bounds that have high accuracy for
a limited number of steps of the preferential attachment process
and then use the union bound to combine these short-interval bounds
to arbitrarily long intervals.

Our concentration bounds in Theorem~\ref{thm:degree-bounds} yield insights on
how preferential attachment graphs form over time: Initially they are in a
chaotic state but after a while central
hubs will emerge. These hubs are then very likely to remain the central hubs
at all times in the future.
It is very unlikely that vertices which have been added after the early
phase will outgrow the central hubs from the early phase.

\paragraph{\Purns.}
In Section~\ref{sec:bounds}, we prove concentration
bounds for degrees in preferential attachment graphs
directly, using only Chernoff arguments.
But one can also model
the behavior of degrees
over time using so called \emph{\Purn processes}.
\Purns can be used to model many interesting biological
or physical processes such as the spreading of diseases or mixing behavior of
particles~\cite{flajolet2006some}. There is a rich literature on probability
distributions of \Purns, both exact and in the
limit~\cite{flajolet2006some,mahmoud2008polya,morcrette2012fully}.
In Section~\ref{sec:urns} we establish a tight connection
between degrees in preferential attachment graphs
and certain \Purn processes and try to
use established results from \Purn theory to compute the
probability distribution of degrees in preferential attachment graphs.
Using \Purns to simulate random graph models is an approach that is not
yet well explored, but there are some
results~\cite{berger2014asymptotic,farczadi2014degree}

At first, we observe in Lemma~\ref{lem:urns_and_degrees} that the
degree of a vertex in a preferential attachment graph follows
exactly the same probability distribution as a certain \Purn process.
We then reinterpret a result from \Purn theory~\cite{flajolet2006some}
as a closed formula for the degree distribution of the first vertex in a
preferential attachment graph.
We use this result to obtain tight bounds
for the degree distribution of this vertex (see Proposition~\ref{prop:easy_case}).
These bounds appear to be previously
unknown in the context of preferential attachment graphs.
We believe that this strong connection between \Purns and preferential
attachment graphs could help to simplify and strength
the bounds presented in Section~\ref{sec:bounds} and may lead to further
structural or algorithmic results for preferential attachment graphs.
We finish Section~\ref{sec:urns} by giving some partial results
for improving the bounds from Section~\ref{sec:bounds}.

\section{Preliminaries}
\label{sec:prelim}
In this work we will denote probabilities by $\P[*]$ and expectation by $\E[*]$. 
We use common graph theory notation~\cite{diestel}.

\subsection{The Preferential Attachment Graph Model}
\label{sec:prelim:model}

We consider a random graph model to be a sequence of probability distributions.
For every $n \in \bf N$ a random graph model describes
a probability distribution on graphs with $n$ vertices.
In this work we focus on the \emph{preferential attachment random graph model}
which we describe in this subsection.
It
has been ambiguously defined in the
original article by \Barabasi and Albert~\cite{barabasi1999emergence}.
The model generates random graphs by iteratively inserting
new vertices and edges.
It depends on a
parameter, usually denoted by $m$, which indicates the number of edges
attached to a newly created vertex. We follow the rigorous
definition of \Bollobas and Ricordan~\cite{Bollob}: For a fixed parameter
$m$ the random process is
defined by starting with a single vertex and iteratively adding vertices,
thereby constructing a sequence of graphs $G_m^1,G_m^2,\dots$,$G_m^t$,  where
$G_m^t$ has $t$ vertices and $mt$ edges. We define $\D[t][m][v]$ to be the
degree of vertex $v$ in the graph $G_m^t$. The random process for $m=1$ works
as follows. A random graph is started with one vertex $v_1$ that has exactly
one self-loop. This graph is $G_1^1$. We then define the graph process
inductively: Given $G_1^{t-1}$ with vertex set $\{v_1,\dots,v_{t-1}\}$, we
create $G_1^t$ by adding a new vertex $v_t$ together with a single edge from
$v_t$ to $v_i$, where $i$ is chosen at random from $\{1,\dots,t\}$
with
$$
\P[i=s] =
\begin{cases}
\displaystyle\frac{d^{t-1}_1(v_s)}{2t-1} &\quad 1\leq s\leq t\\[4\jot]
\displaystyle\frac{1}{2t-1} &\quad s=t.
\end{cases}
$$
This means we add an edge to a random vertex with a probability proportional
to its degree at the time.

For $m>1$, the process can be defined by merging sets of $m$ consecutive
vertices in $G_1^{mt}$ to single vertices in $G_m^t$~\cite{Bollobas:2001}. Let
$v_1,\dots,v_{mt}$ be the vertices of $G_1^t$ and $v_1',\dots,v_t'$ be the
vertices of $G_m^t$. Then $v_i' = \{ v_{im+1},\dots,v_{im+m} \}$. The graph
$G_m^t$ is a multi graph. The number of edges between vertices $v_i'$ and
$v_j'$ in $G_m^t$ equals the number of edges between the corresponding set of
vertices in $G_1^{mt}$. Self-loops are allowed. We focus in most of
our calculations on the case $m=1$ and then reduce the case
of arbitrary values for $m$ on this case.

In Section~\ref{sec:bounds} we obtain concentration bounds for
the total degree of a set of vertices $S
\subseteq \{v_1,\dots,v_{t_0}\}$ during the random process. We define the
degree of a set $S$ at time~$t > t_0$ as
$$
\D[t][m][S] = \sum_{v \in S} \D[t][m][v].
$$
For $m=1$ we can explicitly state the probability distribution of $\D[t]$,
conditioned under $\D[t-1]$ with $t > t_0$. We have
$$
\P[\D[t]=x \mid \D[t-1]] =
\begin{cases}
    \displaystyle \frac{\D[t-1]}{2t-1}   &x=\D[t-1]+1 \\[4\jot]
    \displaystyle1-\frac{\D[t-1]}{2t-1} &x=\D[t-1]   \\[4\jot]
    0                     &\text{otherwise}.
\end{cases}
$$

\subsection{Sparsity}
\label{sec:prelim:sparsity}

There are various ways to define the sparsity of a graph. We consider
the concept of nowhere-denseness.
In order to define it, we need the concept of shallow topological minors.
\begin{definition}[Shallow topological minor~\cite{NOdM08}]
\label{def:shallowtopminor+}
    A graph $M$ is an \emph{$r$-shallow topological minor} of~$G$ if $M$ is
    isomorphic to a subgraph $G'$ of~$G$ if we allow the edges of $M$ to be
    paths of length up to $2r+1$ in $G'$.
    We call $G'$ a \emph{model of $M$ in $G$}. For simplicity we assume by
    default that $V(M) \subseteq V(G')$ such that the isomorphism between $M$
    and $G'$ is the identity when restricted to $V(M)$. The vertices $V(M)$
    are called \emph{principal vertices} and the
    vertices $V(G') \setminus V(M)$ \emph{subdivision vertices}.
    The set of all $r$-shallow topological minors of a graph $G$ is denoted by
    $G \topnab r$.
\end{definition}
We define the clique size over all topological minors of $G$ as
$$\omega(G \topnab r) = \max_{H\in G \topnab r }
\omega(H).$$


\begin{definition}[Somewhere-dense~\cite{sparsity}]
  A graph class~$\mathcal{G}$ is somewhere-dense if for all functions $f$ there
  is an $r$ and a $G\in\mathcal{G}$, such that $\omega( G
  \topnab r ) > f(r)$.
\end{definition}

\begin{definition}[Nowhere-dense~\cite{sparsity}]
  A graph class~$\mathcal{G}$ is nowhere-dense if there exists a function $f$,
  such that for all $r$ and all $G\in\mathcal{G}$, $\omega(G \topnab r ) \leq
  f(r)$.
\end{definition}

We note that $f$ can be an arbitrary function but is only allowed to depend on
$r$. After defining what it means for a graph class to be nowhere- or
somewhere-dense, we now transfer these concepts to random graph models.

\begin{definition}[\aas somewhere-dense]\label{def:aas-somewhere-dense}
A random graph model~$\mathcal{G}$ is \aas\\ somewhere-dense if for all functions
$f$ there is an $r$, such that
   $$
   \lim_{n\to\infty} \P[\omega( G_n \topnab r ) > f(r)] = 1
   $$
   where $G_n$ is a random variable modeling a graph of $n$ vertices randomly
   drawn from $\mathcal{G}$.
\end{definition}

\begin{definition}[\aas nowhere-dense]
A random graph model~$\mathcal{G}$ is \aas nowhere-dense if there exists a
function $f$ such that for all $r$
   $$
   \lim_{n\to\infty} \P[\omega( G_n \topnab r ) \leq f(r)] = 1
   $$
   where $G_n$ is a random variable modeling a graph of $n$ vertices randomly
   drawn from $\mathcal{G}$.
\end{definition}
Observe that a graph class is somewhere-dense if and only if it is
not nowhere-dense. The concepts are complementary.
A random graph model, however, can both be \emph{not} \aas somewhere-dense and
\emph{not} \aas nowhere-dense.

\subsection{The \Purn Process}

We now describe a two-color \Purn process. An exhaustive overview over
\Purns can be found in the book by Mahmoud~\cite{mahmoud2008polya}.
The model works as follows:
An urn contains balls of two colors $A$ and $B$.
Initially at time $n=0$ it contains $a_0$ balls of color $A$
and $b_0$ balls of color $B$.
In the next step a ball is chosen uniformly at random
from the urn and its color is observed.
The ball remains in the urn in this step.
Afterwards balls are added to (or removed from)
the urn according to the replacement matrix
$$
\begin{pmatrix}
\alpha & \beta \\
\gamma & \delta
\end{pmatrix}
$$
with $\alpha, \beta, \gamma, \delta \in \bf N$.
If the observed ball is of color $A$, then
$\alpha$ new balls of type $A$ and $\beta$ new balls of type $B$ are placed
into the urn; if the ball is of color $B$, then $\gamma$ new balls of
type $A$ and $\delta$ new balls of type $B$ are placed into the urn.
A negative value for $\alpha, \beta, \gamma, \delta$, denotes that balls have
to be removed instead of added.
In the next step another ball is chosen uniformly random from the urn and
this procedure is repeated multiple times.
One is interested in the distribution of colors in the urn
after some number of $n$ steps.
\begin{definition}
    We define the random variable $A(M,n,a_0,b_0)$ to
    be the number of balls of color $A$ after $n$ steps
    in the \Purn process with replacement matrix $M$
    and initial configuration $(a_0,b_0)$.
    If $M,a_0,b_0$ are clear from the context
    we also write $A_n$ as a shorthand.
\end{definition}

We call an urn \emph{balanced} with balance $\sigma$
if $\alpha+\beta=\gamma+\delta=\sigma$.
Let $s_n$ be the number of balls in an urn at time $n$. For balanced
urns we have $s_n=a_0+b_0+n\sigma=s_0+n\sigma$.
An urn is \emph{additive} if $\alpha,\beta,\gamma,\delta \ge 0$.
An urn is called \emph{triangular} if $\gamma=0$.
For typographic reasons we write the replacement matrix as
$[\alpha,\beta,\gamma,\delta]$ and an urn model as $(M,a_0,b_0)$, where $M$ is
the replacement matrix and $(a_0,b_0)$ is the initial configuration.

\section{Degree Bounds in Preferential Attachment \\ Graphs}
\label{sec:bounds}
The main contribution of this section is to
show that under certain conditions the degree of vertices
is closely centered around their expected value.
Assume a subset $S$
of the first $t$ vertices with a total degree of $\D[t]$ at time $t$.
If $\D[t]$ is high then the total degree of $S$ at time $n$ is
close to its expected value, which is approximately
$\sqrt{n/t}\;\D[t]$~\cite{hofstad1}.
The higher $\D[t]$ is, the closer the degree is centered
at time $t$.  This is formalized in Theorem~\ref{thm:degree-bounds},
which is proven at the end of this section.
This section consists of multiple lemmas that together prove
Theorem~\ref{thm:degree-bounds}.
We separately show upper and lower bounds and then join these bounds together.
These upper and lower bounds are proven by first giving
bounds which hold for a short time interval of time
(Section~\ref{sec:bound-short-interval})
and then extending these bounds for longer intervals
(Section~\ref{sec:bound-long-interval}).

Due to the technical nature of this section, we consider the set $S$
to be fixed and write $D(t)$ as shorthand
for $\D[t]$ to avoid having large formulas as a superscript. We also define
$D(t) := D(\lfloor t \rfloor)$ for $t \in \bf R$.

\subsection{Short-Term Bounds}
\label{sec:bound-short-interval}

In the following we show that for small $\delta$
from time-step $t$ to $(1+\delta)t$ it is very
likely that we increase the degree of the set $S$ by a factor 
of $1 + \delta/2 + O(\delta^2)$.
A lower bound is established in Lemma~\ref{lem:lower-degbound-delta}
and an upper bound follows from Lemma~\ref{lem:upper-bootstrap}
and \ref{lem:upper-degbound-delta}.
\begin{lemma}
\label{lem:lower-degbound-delta}
    Let $0 < \delta < 1$ and $t\ge \frac{2}{\delta^2}$. Then
$$
    \P\Bigl[D\big((1+\delta)t\big)
    \leq \big(1 + \frac{\delta}{2} - 2\delta^2\big) D(t) \bigm| D(t)\Bigr]
    \leq e^{-\frac{1}{16}\delta^3 D(t)}.
$$
\end{lemma}
\begin{proof}
For every $t' \in \R$ $D(t') = D(\lfloor t' \rfloor)$.
For every $t' \in \N$ either $D(t')=D(t'-1)$ or $D(t')=D(t'-1)+1$.
Let $N$ be the number of integers between $t$ and $(1+\delta)t$.
Let $\Delta_i$ with $1 \le i \le N$ be the Bernoulli variable indicating
that $D(\lfloor t \rfloor + i) = D(\lfloor t \rfloor + i - 1)+1$
and $\Delta = \Delta_1+\dots+\Delta_N$.
Then $D(t) + \Delta=D((1+\delta)t)$.
Furthermore
$$
\P[\Delta_i=1 \mid \Delta_1,\dots,\Delta_{i-1},D(t)]
= \frac{D(\lfloor t \rfloor + i - 1)}{2(\lfloor t \rfloor + i) - 1}
\ge \frac{D(t)}{2(1+\delta)t}.
$$
Let $X=X_{1}+\cdots+X_{N}$ be the sum of identically distributed Bernoulli
variables with
\begin{equation}\label{eqn:prob-xi}
\P[X_i=1] = \frac{D(t)}{2(1+\delta)t}.
\end{equation}
We consider two experiments: The first game is $N$ tosses of a fair coin. The
second one is $N$ tosses of a biased coin, where the probability that the
$i$th coin
comes up head depends on the outcome of the previous coins but always is at
least $1/2$. Obviously, the probability of at least $s$ heads in the second experiment
is at least as high as the probability of at least $s$ heads in the first experiment. 
The same argument implies
\begin{equation}
    \label{equ:lower-xi>delta}
    \P[\Delta \le s\mid D(t)] \leq \P[X \le s\mid D(t)].
\end{equation}
With $t \ge \frac{2}{\delta^2}$ and (\ref{eqn:prob-xi}) we get
$$
N \ge \delta t - 1 \ge (\delta - \frac12 \delta^2) t
$$
and
\begin{equation}
\label{eqn:expbigger}
E[X \mid D(t)]
= N \P[X_i=1\mid D(t)]
\ge \frac{(\delta - \delta^2/2) D(t)}{2(1+\delta)}.
\end{equation}
In contrast to $\Delta$, we can directly apply Chernoff bounds
to~$X$:
\begin{equation}
\label{eqn:chernov-prob-bound}
 \P\Big[X \le (1 - \delta) E\big[X\mid D(t)\big] \Bigm| D(t)\Big]
 \leq e^{-\frac{1}{2} \delta^2 E[X | D(t)]}.
\end{equation}
Combining the above inequality with (\ref{eqn:expbigger}),
(\ref{equ:lower-xi>delta}) and (\ref{eqn:chernov-prob-bound}) yields
\begin{align}\label{eqn:lower-deltabound}
    &\P\Bigl[\Delta \leq
    \frac{(1-\delta)(\delta-\delta^2/2) D(t)}{2(1+\delta)}
    \Bigm| D(t)\Bigr]
    \stackrel{\text{(\ref{eqn:expbigger})}}{\leq} \P\Bigl[\Delta \leq (1-\delta)E[X \mid D(t)]
    \Bigm| D(t)\Bigr]\nonumber\\
    &\stackrel{\text{(\ref{equ:lower-xi>delta})}}{\leq} \P\Bigl[X \leq (1-\delta)E[X \mid D(t)]
    \Bigm| D(t)\Bigr]
    \stackrel{\text{(\ref{eqn:chernov-prob-bound})}}{\leq} e^{-\frac{1}{2} \delta^2 E[X | D(t)]}
    \stackrel{\text{(\ref{eqn:expbigger})}}{\leq}e^{-\frac{\delta^3 - \delta^4/2}{4(1+\delta)} D(t)}\\
    &\leq e^{-\frac{1}{16}\delta^3 D(t)}.
\end{align}
The left and right hand side of the following inequality are identical
for $\delta=0$ and their difference is increasing for
$0\leq\delta\leq1$.  Hence for $0\le\delta\le1$
\begin{equation}
    \label{eqn:lower-taylor}
    \frac{(1-\delta)(\delta - \delta^2/2)}{2(1+\delta)} \ge
    \frac{\delta}{2} - 2\delta^2.
\end{equation}
Combining (\ref{eqn:lower-deltabound}), (\ref{eqn:lower-taylor}), and
$D((1+\delta)t) = \Delta + D(t)$ yields
\begin{multline*}
\P\bigl[ \Delta \le (\delta/2 - 2\delta^2) D(t) \bigm| D(t)\bigr]
= \P\bigl[D((1+\delta)t)
  \le (1 + \delta/2 - 2\delta^2) D(t) \bigm| D(t)\bigr]\\
\leq e^{-\frac{1}{16}\delta^3 D(t)}.
\end{multline*}
\end{proof}

\begin{lemma}
\label{lem:upper-bootstrap}
Let $0<\delta<1$, $t\ge \frac{2}{\delta^2}$, $k > 0$
and $\Delta = D ((1+\delta) t) - D (t)$.
Then
$$
\P\Bigl[\Delta\geq (\delta/2+\delta^2/2) (D(t)+k)
\bigm| D(t),\Delta\leq k\Bigr]
\leq e^{-\frac18\delta^3 D(t)}.
$$
\end{lemma}
\begin{proof}
This proof is similar to the one of Lemma~\ref{lem:lower-degbound-delta}.
Let $N$ be the number of integers between $t$ and $(1+\delta)t$.
Let $\Delta_i$ with $1 \le i \le N$ be the Bernoulli variable indicating
that $D(\lfloor t \rfloor + i) = D(\lfloor t \rfloor + i - 1)+1$.
Then $\Delta = \Delta_1+\dots+\Delta_N$.
Furthermore
$$
\P[\Delta_i=1 \mid \Delta_1,\dots,\Delta_{i-1},\Delta\le k, D(t)]
= \frac{D(\lfloor t \rfloor + i + 1)}{2(\lfloor t \rfloor + i) - 1}
\le \frac{D(t)+k}{2t}.
$$
Let $X=X_{1}+\cdots+X_{N}$ be the
sum of identically distributed Bernoulli variables
with
\begin{equation}\label{eqn:prob-xi-upper}
\P[X_i=1] = \frac{D(t)+k}{2t}.
\end{equation}
With a similar argument as in the proof of
Lemma~\ref{lem:lower-degbound-delta} we have
\begin{equation}
    \label{equ:upper-xi>delta}
    \P[\Delta \ge s\mid D(t)] \leq \P[X \ge s\mid D(t)].
\end{equation}
With $t \ge \frac{2}{\delta^2}$ and (\ref{eqn:prob-xi-upper}) we get
$$
N \le \delta t + 1 \ge (\delta + \frac12 \delta^2) t
$$
and
\begin{equation}
\label{eqn:expbigger-upper}
E[X \mid D(t)]
= N \P[X_i=1\mid D(t)]
\ge \frac{(\delta + \delta^2/2) (D(t) + k)}{2}.
\end{equation}
Chernoff bounds applied on $X$ yield:
\begin{equation}
\label{eqn:chernov-prob-bound-upper}
 \P\Big[X \ge (1 + \delta) E\big[X\mid D(t)\big] \Bigm| D(t)\Big]
 \leq e^{-\frac{1}{2} \delta^2 E[X | D(t)]}.
\end{equation}
For $0\leq\delta\leq1$ we have
\begin{equation}
    \label{eqn:upper-taylor}
    \frac{(1+\delta)(\delta + \delta^2/2)}{2} \le
    \frac{\delta}{2} + \frac{\delta^2}{2}.
\end{equation}
Combining (\ref{eqn:upper-taylor}), (\ref{eqn:expbigger-upper}),
(\ref{equ:upper-xi>delta}) and (\ref{eqn:chernov-prob-bound-upper}) yields
\begin{align*}
    &\P\Bigl[\Delta \ge
    (\delta/2 + \delta/2) (D(t)+k)
    \Bigm| D(t)\Bigr]
    \stackrel{\text{(\ref{eqn:expbigger-upper})}}{\leq} \P\Bigl[\Delta \ge (1+\delta)E[X \mid D(t)]
    \Bigm| D(t)\Bigr]\nonumber\\
    &\stackrel{\text{(\ref{equ:upper-xi>delta})}}{\leq} \P\Bigl[X \ge (1+\delta)E[X \mid D(t)]
    \Bigm| D(t)\Bigr]
    \stackrel{\text{(\ref{eqn:chernov-prob-bound-upper})}}{\leq} e^{-\frac{1}{2} \delta^2 E[X | D(t)]}
    \stackrel{\text{(\ref{eqn:expbigger-upper})}}{\leq}e^{-\frac{\delta^3 - \delta^4/2}{4} D(t)}\\
    &\leq e^{-\frac{1}{8}\delta^3 D(t)}.
\end{align*}
\end{proof}

\begin{lemma}
\label{lem:upper-degbound-delta}
Let $0 < \delta \le \frac{1}{e^2}$ and $t\ge \frac{2}{\delta^2}$.
Then
$$
\P\Bigl[D((1+\delta)t) \ge ( 1+ \delta/2 + 2\delta^2 ) D(t) \Bigm| D(t)\Bigr]
\le \ln(e2t) e^{-\frac18 \delta^3D(t)}
$$
\end{lemma}
\begin{proof}
By setting $\Delta=D((1+\delta) t)-D(t)$, $\beta=(\delta+\delta^2)/2$,
$p=e^{\frac18 \delta^3 D(t)}$ and using Lemma~\ref{lem:upper-bootstrap} we get
\begin{equation}
\label{equ:bootstrap}
\P\bigl[\Delta\geq\beta D(t)+\beta k\bigm| D(t),\Delta\leq k\bigr]\leq p.
\end{equation}
We define bounds $ k_i$ recursively by $ k_0=\delta t$ and
$ k_{i+1}=\beta D(t)+\beta k_i$ for $i\geq0$.
Let us proof by induction on $i$ that
\begin{equation}
\label{equ:boot}
\P[\Delta\geq k_i\mid D(t)]\leq ip.
\end{equation}
The base case $i=0$ follows from $\Delta\leq k_0$ being always true. Let
us now assume that $i\geq0$ and split the left-hand side of
(\ref{equ:boot}) into
\begin{align*}
\P[\Delta\geq k_{i+1}\mid D(t)]=
&\P[\Delta\geq k_{i+1}\mid D(t),\Delta\geq k_i]
  \P[\Delta\geq k_i\mid D(t)]\\
+&\P[\Delta\geq k_{i+1}\mid D(t),\Delta< k_i]
  \P[\Delta< k_i\mid D(t)].
\end{align*}
By induction we know that $\P[\Delta\geq k_i\mid D(t)]\leq ip$ and by
(\ref{equ:bootstrap}) we know that $\P[\Delta\geq k_{i+1}\mid
D(t),\Delta< k_i]\leq p$. Together this is at most $(i+1)p$. This
concludes the proof of (\ref{equ:boot}).
Solving the linear recurrence relation $ k_i$ yields
$$
 k_i=\beta D(t)+\beta^2 D(t)+\cdots+\beta^i D(t)+\beta^i k_0
\le \frac{\beta}{1-\beta}D(t) + \beta^i  k_0.
$$
With $0 < \delta < 1/e^2$ we have
$$
\frac{\beta}{1-\beta} \le \delta/2 + \delta^2 3/2
$$
and
\begin{align*}
\beta^{\ln(e2t)} k_0
&\le \Bigl(\frac{\delta + \delta^2}{2}\Bigr)^{\ln(e2t)}  k_0
\le \delta^{\ln(e2t)} \delta t
\le \delta^2 \delta^{\ln(2t)} t
\le \delta^2 (2t)^{\ln(\delta)} t
\le \delta^2 (2t)^{-2} t\\
&\le \frac{\delta^2}{2}
\end{align*}
Therefore
\begin{equation}\label{eqn:e2t}
 k_{\ln(e2t)}
\le \frac{\beta}{1-\beta}D(t) +  k_0^{\ln(e2t)}
\le (\delta/2 + 2\delta^2) D(t).
\end{equation}
Combining (\ref{equ:boot}) and (\ref{eqn:e2t}) gives us
with $i=\ln (e2t)$ and $D((1+\delta)t) = \Delta + D(t)$
\begin{align*}
    &\P\Bigl[D((1+\delta)t) \ge ( 1+ \delta/2 + 2\delta^2 ) D(t) \Bigm| D(t)\Bigr]\\
    =& \P\Big[\Delta + D(t) \geq (1 + \delta/2 + 2\delta^2) D(t) \Bigm| D(t)\Big]
    =\P\Big[\Delta \geq \delta/2 + 2\delta^2 D(t) \Bigm| D(t)\Big]\\
    \leq& \P\Big[\Delta \geq  k_{\ln (e2t)} \Bigm| D(t)\Big]
    \leq \ln (e2t)e^{\frac18 \delta^3 D(t)}.
\end{align*}
\end{proof}

\subsection{Long-Term Bounds}
\label{sec:bound-long-interval}


In the previous subsection we established bounds
for a small interval from step $t$ to step $(1+\delta)t$
with an error of order $\delta^2$.
In this subsection we combine these bounds into long-term bounds.
We get these bounds by defining positions
$t_0=t$ and $t_{k+1}=(1+\delta_k)t_k$ with $k \in \bf N$
and using the union bound to guarantee that for each interval
from time $t_k$ to $t_{k+1}$ the short-term bounds hold.
The choice of $\delta_k$ is of high importance for the success
of this strategy.
It turns out that we need the product
$\prod_{k=1}^{\infty}(1+\delta_k)$ to diverge, but
the error $\prod_{k=1}^{\infty}(1+\delta_k^2)$ to converge.
Eventually, we settled for $\delta_k = \varepsilon/k^{2/3}$,
which satisfies both conditions.

Lemma~\ref{lem:monotone-increasing-function}
and Lemma~\ref{lem:monotone-increasing-function-upper}
bridge the gap between the bounds for the recursively
defined small intervals and the expected growth of
the degree of a vertex over a longer period of time.
These lemmas state that if the degree differs by a factor of
$(1\pm\varepsilon)$ from its expected value
then there has been one short interval where
the allowed error of order $\delta^2$ has been exceeded.
The lemmas do not directly mention degrees of vertices
but a monotone increasing function $f(t)$. This function
will later be replaced with the degree of a vertex.
Lemma~\ref{lem:bound-single-delta_k-step} then uses
the short-term bounds from
Lemma~\ref{lem:lower-degbound-delta} and Lemma~\ref{lem:upper-degbound-delta}
to give bounds for each interval $t_k$ to $t_{k+1}$.
Finally, Lemma~\ref{lem:degree-bounds-m-one} combines
Lemma~\ref{lem:monotone-increasing-function} and
Lemma~\ref{lem:bound-single-delta_k-step}
into bounds for preferential attachment graphs with parameter $m=1$
and Theorem~\ref{thm:degree-bounds} generalizes
this result for arbitrary $m$.

\begin{lemma}\label{lem:monotone-increasing-function}
Let $0 < \varepsilon \le 1/8$, $t>0$, and $f\colon \R \to
\R$ be an increasing function.  For every $k\in\N$ let
$\delta_k = \frac{\varepsilon}{k^{2/3}}$,
$h_{k}=\prod_{i=1}^{k-1} (1+\delta_i)$, and
$c_{k}=\prod_{i=1}^{k-1} (1+\frac12 \delta_i - 2\delta^2_i)$.

If there is an $n \in \N$, such that $t<n$ and $f(n) < (1-\varepsilon)
\sqrt{\frac{n}{t}} f(t)$, then there is a $k \in \bf N$ such that
$f((1+\delta_k)h_{k}t) < (1+\frac12\delta_k-2\delta_k^2) f(h_kt)$ and $f(h_{k}
t) \ge c_{k} f(t)$.
\end{lemma}
\begin{proof}

Consider any $n \in \bf N$, $n\geq t$. Let $k(n) \in \bf N$ be the maximal
value such that $h_{k(n)} t \le n$. Then $\frac{n}{1+\delta_{k(n)}} \le
h_{k(n)} t$, because of the maximality of~$k(n)$.
Notice that
$$
(1-\varepsilon) \sqrt{\frac{n}{t}}
\le e^{-\frac12 \varepsilon} e^{-\frac12 \varepsilon}\sqrt{\frac{n}{t}}
=   e^{-\frac12 \varepsilon} \sqrt{\frac{n}{te^{\varepsilon}}}
\le   e^{-\frac12 \varepsilon} \sqrt{\frac{n}{t(1+\delta_{k(n)})}}
\le e^{-\frac12 \varepsilon} \sqrt{h_{k(n)}}
$$
and for all $k\in\N$
\begin{multline*}
c_k =   \prod_{i=1}^{k-1} (1 + \frac12\delta_i - 2\delta_i^2)
    \ge \prod_{i=1}^{k-1} e^{\frac12\delta_i - 3\delta_i^2}
    \ge \Bigl(\prod_{i=1}^{k-1} e^{\delta_i} \Bigr)^{\frac12}
      \prod_{i=1}^{\infty} e^{-\frac{3\varepsilon^2}{i^{4/3}}}\\
    \ge \Bigl(\prod_{i=1}^{k-1} (1 + \delta_i) \Bigr)^{\frac12}
         e^{-4\varepsilon^2}
    \ge  e^{-\frac12 \varepsilon} \sqrt{h_k}.
\end{multline*}
Combining the upper two inequalities gives us
$(1-\varepsilon)\sqrt{n/t}\leq c_k$.
We assumed $f(n) < (1-\varepsilon) \sqrt{\frac{n}{t}} f(t)$.
Monotonicity of $f$ yields
$$
f(h_{k(n)}t) \le f(n) < (1-\varepsilon) \sqrt{\frac{n}{t}} f(t) \le c_{k(n)} f(t).
$$
Let $J=\{\,j\geq0\mid f(h_{j+1}t)<c_{j+1}f(t)\,\}$.  The set $J$ is not empty
because $k(n)-1\in J$ by the equation above.
Furthermore, $0\notin J$ because $h_1=c_1=1$
and therefore $f(h_1t)=f(t)=c_1f(t)$.

Let now $k$ be the minimal value in~$J$.
Then $k>0$, $f(h_{k} t) \ge c_{k} f(t)$, and
$f(h_{k+1} t) < c_{k+1} f(t)$.
At last, we have
\begin{multline*}
f((1+\delta_k) h_kt)
=   f(h_{k+1}t)
< c_{k+1} f(t)\\
=   (1+\frac12 \delta_k-2\delta_k^2) c_k f(t)
\le (1+\frac12 \delta_k-2\delta_k^2) f(h_kt).
\end{multline*}
\end{proof}

\begin{lemma}\label{lem:monotone-increasing-function-upper}
Let $0 < \varepsilon \le 1/40$, $t>0$, and $f\colon \R \to
\R$ be an increasing function.  For every $k\in\N$ let
$\delta_k = \frac{\varepsilon}{k^{2/3}}$,
$h_{k}=\prod_{i=1}^{k-1} (1+\delta_i)$, and
$c_{k}=\prod_{i=1}^{k-1} (1+\frac12 \delta_i + 2\delta^2_i)$.

If there is an $n \in \N$, such that $t<n$ and $f(n) > (1+\varepsilon)
\sqrt{\frac{n}{t}} f(t)$, then there is a $k \in \bf N$ such that
$f((1+\delta_k)h_{k}t) > (1+\frac12\delta_k+2\delta_k^2) f(h_kt)$ and $f(h_{k}
t) \le c_{k} f(t)$.
\end{lemma}
\begin{proof}
This proof is similar to the one of Lemma~\ref{lem:monotone-increasing-function}.
Consider any $n \in \bf N$, $n\geq t$. Let $k(n) \in \bf N$ be the minimal
value such that $n \le h_{k(n)} t$.
Then $h_{k(n)} t \le n(1+\delta_{k(n)})$, because of the maximality of~$k(n)$.
Notice that
$$
(1+\varepsilon) \sqrt{\frac{n}{t}}
\ge e^{\frac14 \varepsilon} e^{\frac12 \varepsilon}\sqrt{\frac{n}{t}}
=   e^{\frac14 \varepsilon} \sqrt{\frac{ne^{\varepsilon}}{t}}
\ge   e^{\frac14 \varepsilon} \sqrt{\frac{n(1+\delta_{k(n)})}{t}}
\ge e^{\frac14 \varepsilon} \sqrt{h_{k(n)}}
$$
and for all $k\in\N$
\begin{multline*}
c_k =   \prod_{i=1}^{k-1} (1 + \frac12\delta_i + 2\delta_i^2)
    \le \prod_{i=1}^{k-1} e^{\frac12\delta_i + 2\delta_i^2}
    = \prod_{i=1}^{k-1} e^{\frac12(\delta_i - \delta_i^2) + \frac52\delta_i^2} \\
    \le \Bigl(\prod_{i=1}^{k-1} e^{\delta_i -\delta_i^2} \Bigr)^{\frac12}
    \prod_{i=1}^{\infty} e^{\frac{\frac52 \varepsilon^2}{i^{4/3}}}
    \le \Bigl(\prod_{i=1}^{k-1} (1 + \delta_i) \Bigr)^{\frac12}
         e^{10\varepsilon^2}
    \le e^{\frac14 \varepsilon} \sqrt{h_k}.
\end{multline*}
Combining the upper two inequalities gives us
$(1+\varepsilon)\sqrt{n/t}\ge c_k$.
We assumed $f(n) > (1+\varepsilon) \sqrt{\frac{n}{t}} f(t)$.
Monotonicity of $f$ yields
$$
f(h_{k(n)}t) \ge f(n) > (1+\varepsilon) \sqrt{\frac{n}{t}} f(t) \ge c_{k(n)} f(t).
$$
Let $J=\{\,j\geq0\mid f(h_{j+1}t)>c_{j+1}f(t)\,\}$.  The set $J$ is not empty
because $k(n)-1\in J$ by the equation above.
Furthermore, $0\notin J$ because $h_1=c_1=1$
and therefore $f(h_1t)=f(t)=c_1f(t)$.
Let now $k$ be the minimal value in~$J$.
Then $k>0$, $f(h_{k} t) \le c_{k} f(t)$, and
$f(h_{k+1} t) > c_{k+1} f(t)$.
At last, we have
\begin{multline*}
f((1+\delta_k) h_kt)
=   f(h_{k+1}t)
> c_{k+1} f(t)\\
=   (1+\frac12 \delta_k+2\delta_k^2) c_k f(t)
\ge (1+\frac12 \delta_k+2\delta_k^2) f(h_kt).
\end{multline*}
\end{proof}

\begin{lemma}\label{lem:bound-single-delta_k-step}
Let $0 < \varepsilon \le 1/40$, $t> \frac{1}{\varepsilon^6}$.
For every $k\in\N$ let
$\delta_k = \frac{\varepsilon}{k^{2/3}}$,
$h_{k}=\prod_{i=1}^{k-1} (1+\delta_i)$,
$c^+_{k}=\prod_{i=1}^{k-1} (1+\frac12 \delta_i + 2\delta^2_i)$ and
$c^-_{k}=\prod_{i=1}^{k-1} (1+\frac12 \delta_i - 2\delta^2_i)$.
Then
\begin{multline*}
\P\Bigl[ D((1+\delta_k) h_{k}t) < (1+\frac12 \delta_k - 2\delta_k^2)
D(h_k t), D(h_{k}t) \ge c^-_kD(t) \bigm| D(t) \Bigr]
\\ \le e^{-\frac{1}{16} \delta_k^3 c^-_k D(t)},
\end{multline*}
\begin{multline*}
\P\Bigl[ D((1+\delta_k) h_{k}t) > (1+\frac12 \delta_k + 2\delta_k^2)
D(h_k t), D(h_{k}t) \le c^+_kD(t) \bigm| D(t) \Bigr]
\\ \le  \ln(e2t) e^{-\frac18 \delta_k^3 c^+_k D(t)}.
\end{multline*}
\end{lemma}
\begin{proof}
At first we focus on the first bound.
By the law of total probability
\begin{multline*}
\P\Bigl[ D((1+\delta_k) h_{k}t) <
    (1+\frac12 \delta_k - 2\delta_k^2) D(h_k t),
    D(h_{k}t) \ge c^-_kD(t) \bigm| D(t) \Bigr] \\
\le \P\Bigl[ D((1+\delta_k) h_{k}t) <
    (1+\frac12 \delta_k - 2\delta_k^2) D(h_k t)
    \bigm| D(h_{k}t) \ge c^-_kD(t) \Bigr].
\end{multline*}
The second line of this equation states the probability
that the degree of a vertex is in the future below a certain threshold under
the condition that it is currently above a certain threshold.
We can bound this probability if we assume that it
currently is not above, but exactly at the threshold:
\begin{multline*}
\P\Bigl[ D((1+\delta_k) h_{k}t) <
    (1+\frac12 \delta_k - 2\delta_k^2) D(h_k t) \bigm|
    D(h_{k}t) \ge c^-_kD(t) \Bigr] \\
\le \P\Bigl[ D((1+\delta_k) h_{k}t) <
    (1+\frac12 \delta_k - 2\delta_k^2) D(h_k t) \bigm|
    D(h_{k}t) = c^-_kD(t) \Bigr].
\end{multline*}
Similarly, the probability that
the degree of a vertex is in the future above a certain threshold under
the condition that it is currently below a certain threshold
can be bounded by assuming that it is currently exactly at the threshold.
It is therefore sufficient to prove the following two bounds
$$
\P\Bigl[ D((1+\delta_k) h_{k}t) < (1+\frac12 \delta_k - 2\delta_k^2)
D(h_k t) \bigm| D(h_{k}t) = c^-_kD(t) \Bigr] 
\le e^{-\frac{1}{16} \delta_k^3 c^-_k D(t)},
$$
$$
\P\Bigl[ D((1+\delta_k) h_{k}t) > (1+\frac12 \delta_k + 2\delta_k^2)
D(h_k t) \bigm| D(h_{k}t) = c^+_kD(t) \Bigr]
\le  \ln(e2t) e^{-\frac18 \delta_k^3 c^+_k D(t)}.
$$
Lemma~\ref{lem:lower-degbound-delta} and~\ref{lem:upper-degbound-delta}
state that if $0 \le \delta_k = e/k^{3/2} \le 1/e^2$ and
$h_k t \ge 2/\delta_k^2$ for every $k$ then these bounds are true.
We observe that for $0 \le \varepsilon \le 1/8$ the first
precondition is always satisfied.
We will finish the proof by showing that $h_k t \ge 2/\delta_k^2$ for every $k$.
Observe that for $0 \le k \le 1$ we have $h_k t \ge 2/\varepsilon^2 \ge 2/\delta_k^2$.
We can therefore assume $k \ge 2$.
First, we need a lower bound for $h_k$
$$
h_k = \prod_{i=1}^{k-1} (1+\frac{\varepsilon}{i^{2/3}})
    \ge \prod_{i=1}^{k-1} e^{\frac{\varepsilon}{i^{2/3}}}
    \ge e^{3\varepsilon(k-1)^{1/3}}
    \ge e^{2\varepsilon k^{1/3}}.
$$
One can show that $e^x/x^4 \ge e^4/256$ for $x > 0$.
We therefore get for $x=2\varepsilon k^{1/3}$
$$
h_k
\ge e^{2\varepsilon k^{1/3}}
= \frac{e^{2\varepsilon k^{1/3}}}{(2\varepsilon k^{1/3})^4}
        \frac{16 \varepsilon^6}{\delta_k^2}
\ge \frac{e^x}{x^4} \frac{16 \varepsilon^6}{\delta_k^2}
\ge \frac{e^4}{256} \frac{16 \varepsilon^6}{\delta_k^2}
\ge \frac{2\varepsilon^6}{\delta_k^2}
$$
Since $t \ge \frac{1}{\varepsilon^6}$ it follows that $h_k t \ge \frac{2}{\delta_k^2}$.
\end{proof}
\begin{lemma}\label{lem:degree-bounds-m-one}
    For $0 < \varepsilon \le 1/40$ and $\frac{1}{\varepsilon^6}<t \in \bf N$ we have
    \begin{multline*}
        \P\Bigl[  (1-\varepsilon) \sqrt{\frac{n}{t}} D(t) < D(n)
            < (1+\varepsilon) \sqrt{\frac{n}{t}} D(t) \text{ for all $n \ge t$ }
            \Bigm| D(t) \Bigr] \\ \ge 1 - 2\ln(e2t) \varepsilon^{-6} \exp\bigl(-\varepsilon^{15} 10^{-24} D(t) \bigr)
    \end{multline*}
\end{lemma}
\begin{proof}
Observe that
\begin{multline*}
    \P\Bigl[  (1-\varepsilon) \sqrt{\frac{n}{t}} D(t) < D(n)
        < (1+\varepsilon) \sqrt{\frac{n}{t}} D(t) \text{ for all $n \ge t$ }
    \Bigm| D(t) \Bigr] \\ \ge 1 - (p^+ + p^-)
\end{multline*}
with
$$
    p^- := \P\Bigl[ D(n) < (1-\varepsilon) \sqrt{\frac{n}{t}} D(t) 
    \text{ for some } n \ge t \Bigm| D(t)\Bigr] ~
$$
$$
    p^+ := \P\Bigl[ D(n) > (1+\varepsilon) \sqrt{\frac{n}{t}} D(t) 
    \text{ for some } n \ge t \Bigm| D(t)\Bigr].
$$
We proceed by finding upper bounds for $p^+$ and $p^-$.
For $k \in \N$
let $\delta_k = \frac{\varepsilon}{k^{2/3}}$,
$h_{k}=\prod_{i=1}^{k-1} (1-\delta_i)$,
$c^-_{k}=\prod_{i=1}^{k-1} (1-\frac12 \delta_i - 2\delta^2_i)$ and
$c^+_{k}=\prod_{i=1}^{k-1} (1-\frac12 \delta_i + 2\delta^2_i)$.
Every function $f(t) : \R \to \R$
which is a realization of the random variables $D(t)$ is monotone increasing.
It follows using Lemma~\ref{lem:monotone-increasing-function},
Lemma~\ref{lem:monotone-increasing-function-upper},
the union bound over all possible choices of $k$,
and Lemma~\ref{lem:bound-single-delta_k-step} that
\begin{multline*}
      p^- \le \sum_{k=0}^{\infty} \P\Bigl[ D((1+\delta_k) h_{k}t)
        < (1+\frac12 \delta_k - 2\delta_k^2) D(h_k t), D(h_{k}t) \ge c^-_kD(t) \Bigm| D(t) \Bigr] \\
        \le \sum_{k=0}^{\infty} e^{-\frac{1}{16} \delta_k^3 c^-_k D(t)},
\end{multline*}
\begin{multline*}
      p^+ \le \sum_{k=0}^{\infty} \P\Bigl[ D((1+\delta_k) h_{k}t)
        > (1+\frac12 \delta_k + 2\delta_k^2) D(h_k t), D(h_{k}t) \le c^+_kD(t) \Bigm| D(t) \Bigr] \\
        \le  \sum_{k=0}^{\infty} \ln(e2t) e^{-\frac18 \delta_k^3 c^+_k D(t)}.
\end{multline*}
We finish this proof with a longer calculation.
\begin{align*}
    & p^+ + p^- \le (\ln(e2t)+1) \sum_{k=0}^{\infty}  e^{-\frac{1}{16} \delta_k^3 c^-_k D(t)} \\
    & =  \ln(15t)  \sum_{k=0}^{\infty} \exp\bigl(-\frac{\varepsilon^3}{16k^{2}}
          \prod_{i=1}^{k-1} (1+\frac{\varepsilon}{4i^{2/3}} - \frac{\varepsilon^2}{16i^{4/3}}) D(t) \bigr) \\
    &=   \ln(15t) \sum_{k=0}^{\infty} \exp\bigl(-\frac{\varepsilon^3}{16k^{2}}
          e^{\frac13\varepsilon k^{1/3}} D(t) \bigr) \\
    &\le \ln(15t) \Bigl( \sum_{k=0}^{\varepsilon^{-6}} \exp\bigl(-\frac{\varepsilon^{15}}{16} D(t) \bigr)
          +\sum_{k=\varepsilon^{-6}+1}^{\infty} \exp\bigl(-\frac{\varepsilon^3}{16k^{2}} e^{\frac13 k^{1/6}} D(t) \bigr)\Bigr) \\
    &\le \ln(15t) \Bigl( \varepsilon^{-6} \exp\bigl(-\frac{\varepsilon^{15}}{16} D(t) \bigr)
          +\sum_{k=\varepsilon^{-6}+1}^{\infty} \exp\bigl(- \varepsilon^3 k 10^{-23} D(t) \bigr)\Bigr) \\
    &\le \ln(15t) \varepsilon^{-6} \exp\bigl(-\varepsilon^{15} 10^{-24} D(t) \bigr)
\end{align*}
\end{proof}
%
%
The last step in this section is to
generalize the previous lemma for
different parameters $m$ of the preferential attachment process.
\begin{theorem}\label{thm:degree-bounds}
Let $0 < \varepsilon \le 1/40$, $t,m \in \N$,
    $t > \frac{1}{\varepsilon^6}$ and
    $S \subseteq \{v_1, \dots, v_t\}$.
    Then
    \begin{multline*}
        \P\Bigl[  (1-\varepsilon) \sqrt{\frac{n}{t}} \D[t][m] < \D[n][m]
            < (1+\varepsilon) \sqrt{\frac{n}{t}} \D[t][m] \text{ for all $n \ge t$ }
        \Bigm| \D[t][m] \Bigr] \\ \ge
        1 -  \ln(15t) e^{-\varepsilon^{O(1)}\D[t][m]}.
    \end{multline*}
\end{theorem}
\begin{proof}
As stated in the introduction, we can simulate $G_m^n$ via $G_1^{mn}$, by
merging every $m$ consecutive vertices into a single one.
Let $G_m^n$ be a graph with vertices $V = \{ v_1,\dots,v_n \}$.
We can assume that this graph has been constructed from a graph
$G_1^{mn}$ with vertex set $V' = \{ v_1',\dots,v_{mn}' \}$ by
merging the vertices $v_{i+1}',\dots,v_{i+m}'$ into vertex $v_i$ for $1\le i \le n$.
Let $S' \subseteq V'$ be the set of vertices in $G_1^{mn}$ which are merged into $S$.
Since the graph allows multi-edges it is easy to see
that $\D[n][m][S] = \D[mn][1][S']$.
This means that $\D[n][m][S]$ and $\D[mn][1][S']$ have the same probability distribution.
Lemma~\ref{lem:degree-bounds-m-one} gives with $\D[mn][1][S'] = D(mn)$
the following statement
\begin{multline*}
    \P\Bigl[  (1-\varepsilon) \sqrt{\frac{n}{t}} \D[mt][1][S'] < \D[mn][1][S']
        < (1+\varepsilon) \sqrt{\frac{n}{t}} \D[mt][1][S'] \text{ for all $n \ge t$ }
    \Bigm| \D[t][m] \Bigr] \\ \ge 1 - \ln(15t) \varepsilon^{-6} \exp\bigl(-\varepsilon^{15} 10^{-24} \D[mt][1][S'] \bigr)
\end{multline*}
which proves the claim with $\D[n][m][S] = \D[mn][1][S']$.
\end{proof}


\section{Preferential Attachment Graphs and \\ \Purns}
\label{sec:urns}
In this section we consider a two-color \Purn process
and its relation to the probability distribution
of degrees in preferential attachment graphs.
At first, we observe in Lemma~\ref{lem:urns_and_degrees}
that the degree of vertices
in a preferential attachment graph follows exactly
the same probability distribution as the number
of balls of one color in a certain \Purn process.
This means that concentration results for distributions
in \Purns can be easily restated as concentration results
for degrees of vertices.
Our concentration bounds presented in Section~\ref{sec:bounds}
were obtained using no other methods than Chernoff bounds.
By using the rich set of existing results on \Purns
it might be possible to dramatically improve
these results.
This section builds on results by
Flajolet, Dumas and Puyhaubert~\cite{flajolet2006some}.

The three authors give a closed expression for the color
distribution for a specific urn process
with fixed initial configuration after $n$ steps
(see Proposition~\ref{prop:easy_case}).
We use Lemma~\ref{lem:urns_and_degrees} to
rephrase this result as a closed expression for
the exact probability distribution of the degree
of the first vertex in a preferential attachment graph of size $n$.
Knowing the exact distribution, we then give asymptotically
tight concentration bounds for the degree of the first vertex
in Lemma~\ref{lem:first_vertex_bound}.
We observe that $\P[d_1^n(v_1) > c \sqrt{n}]\leq e^{-c^2/4}$
for any $c,n \in \bf N$.
Our results in Section~\ref{sec:bounds} did not
yield unconditional concentration bounds
for the degree of the first vertex.

Unfortunately, no one seems to know
a closed expression for the degree of arbitrary
vertices under the condition that they have
a certain degree at some earlier time, i.e., probabilities
of the form $\P[\D[n] = k \mid \D[t] = l]$.
The results by Flajolet, Dumas and Puyhaubert~\cite{flajolet2006some}
can be restated as expressions for such probabilities.
These are, however, not closed expressions as they contain an alternating sum.
The summands of this alternating sum are very large
but alternate in sign such that they almost completely
cancel each other out.
This means it is very hard to get an approximate
closed form for this sum. Good approximations
for each summand may still lead to bad approximations
for the whole sum because of cancellation.
We were unable to use these results to get an approximation
for probabilities of the form $\P[\D[n] = k \mid \D[t] = l]$.
However, we were able to replace
the alternating sum with a non-alternating sum.
For each probability of the form $\P[\D[n] = k \mid \D[t] = l]$
we discovered an equivalent sum where
each summand is non-negative.
This result is stated in Lemma~\ref{lem:non_alternating}.
This sum might now be easier to approximate
because no summands cancel each other out and
a good approximation for each summand could
lead to a good approximation for the whole sum.
We believe that this result can help improving
the bounds presented in Section~\ref{sec:bounds}.
We also observed, using computer algebra tools, that
probabilities of the form $\P[\D[n] = k \mid \D[t] = l]$
can be expressed in terms of generalized hypergeometric functions.
There exists a rich set of tools for analyzing
generalized hypergeometric functions~\cite{graham1994concrete},
which also might lead to improved bounds.

\subsection{Relation Between \Purns and Degrees in \\ Preferential Attachment Graphs}

We now observe that the degree of a vertex of a preferential attachment graph
follows the same probability distribution as the color $A$
in the balanced triangular urn of the form
$$
M =\begin{pmatrix}1 & 1 \\0 & 2\end{pmatrix}.
$$
We fix a vertex $v$ of the preferential attachment process.
The degree of $v$ is represented by the number of balls of color $A$ and
the degree of all other vertices is represented by the balls of color $B$.
In each step the total degree of the graph grows by two.
Either the degree of $v$ grows by one and the degree
of the vertex attaching to $v$ grows by one (first line of $M$) or
the degree of $v$ stays the same and the degree of two other
vertices grows by one (second line of $M$).
This relationship is formalized in the following lemma.
\begin{lemma}
\label{lem:urns_and_degrees}
    Let $M=[1,1,0,2]$.
    Let $n \ge t$. For every $d,k \in \bf N$
    \begin{align*}
        \P\big[\D[n][1][v_t]=k\big] &= \P\big[A(M,n-t+1,1,2t-2) = k)\big]
        \\[5pt]
        \P\big[\D[n]=k \bigm| \D[t] =d\big] &= \P\big[A(M,n-t,d,2t-1-d) = k)\big].
    \end{align*}
\end{lemma}
\begin{proof}
    We start with the first equality.
    The random variables $\D[n][1][v_t]$
    and $A(M,n-t+1,1,2t-2)$ both describe how a quantity (the degree
    of a vertex or the number of balls of color $A$) behaves
    over time.
    Initially, this quantity is one and then there are
    $n-t+1$ rounds in which it may be increased by one.
    In each round either two edges or two balls are added.
    Therefore in both random processes the probability that
    the quantity increases in the $i$th round is $x/(2t-1+2i)$,
    assuming it currently is~$x$.
    It follows that both random variables follow
    the same distribution.
Intuitively, one can visualize this equality as follows. The color $A$
refers to the degree of the first vertex, color $B$ to the sum of the
other degrees and $s_n$ to the sum of all degrees at time $n$. We can think of
it as at time $0$ there is a single vertex with degree one, which represents a
stub and at the next time step this gets extended to a self-loop and a new
vertex with a stub is inserted. At the next step this stub is either extended
to an edge to $v_1$ or to a self-loop to $v_2$.

    Now consider the second equality.
    If we assume that $\D[t]=d$ then there are only $n-t$
    rounds after the $t$th round
    in which the degree of the set $S$ may increase.
    The probability that the degree or the number of balls of color $A$
    increases in the $i$th round is again $x/(2t-1+2i)$,
    assuming it currently is $x$.
    It follows that
    $\D[n]$ under the condition $\D[t]=d$ and $A(M,n-t,d,2t-1-d)$
    follow the same distribution.
\end{proof}

\subsection{Tight Bounds for the First Vertex in Preferential \\ Attachment Graphs}

Due to Flajolet, Dumas and Puyhaubert~\cite{flajolet2006some}, we get
the following closed form expression by solving the urn model
$([1,1,0,2],1,0)$.
\begin{proposition}[\cite{flajolet2006some}]
\label{prop:easy_case}
Given the urn model $([1,1,0,2],1,0)$, we get
$$
\P[A_n=k]=\frac{k-1}{n} 2^{k-1} \frac{{2n-k \choose n-1}}{{2n \choose n}}.
$$
\end{proposition}
By Lemma~\ref{lem:urns_and_degrees}, this yields a formula for
$\P[\D[n][1][v_t]=k]$.
We can directly derive asymptotically tight bounds for the
degree of the first vertex.

\begin{lemma}
\label{lem:first_vertex_bound}
Let $n,c \in \bf N$. Then
$\P[d_1^n(v_1) > c \sqrt{n}]\leq e^{-\frac{c^2}{4}}$.
This bound is asymptotically tight in $n$.
\end{lemma}
\begin{proof}
Let $M=[1,1,0,2]$, $a_0=1$ and $b_0=0$.
By Lemma~\ref{lem:urns_and_degrees} we have
$$
    \P[d_1^n(v_1) > c \sqrt{n}] = \sum_{i=\lceil c \sqrt{n} \rceil}^\infty \P[A_n = i].
$$
One can use Proposition~\ref{prop:easy_case}
and computer algebra tools to verify that
$$
    \lim_{n\to\infty} \sum_{i=\lceil c \sqrt{n} \rceil}^\infty
    \P[A_n = i] = e^{-\frac{c^2}{4}}
$$
and that the limit is approached from below.
\end{proof}
From Lemma~\ref{lem:first_vertex_bound} we can follow a special case.
\begin{lemma}
\label{lem:first_vertex_bound_special}
For every $\varepsilon>0$ there exists an $N_\varepsilon$ such that for all
$N_\varepsilon<n\in \bf N$ $\P[d_1^n(v_1) \leq \varepsilon \sqrt{n}]
\geq \frac{1}{n}$.
\end{lemma}
\begin{proof}
We use Lemma~\ref{lem:first_vertex_bound}:
$$
\P[d_1^n(v_1) \leq c \sqrt{n}] = 1 -\P[d_1^n(v_1) > c \sqrt{n}] > 1 -
e^{-\frac{c^2}{4}}
$$
We set $c=2\sqrt{\log(\frac{n}{n-1})}$. For large
enough $n$ we have $c\leq\varepsilon$ for every $\varepsilon>0$. Therefore
for every $\varepsilon>0$ there exists an $N_\varepsilon$ such that for all
$n>N_\varepsilon$
$$
\P[d_1^n(v_1) \leq \varepsilon \sqrt{n}]\geq \P\Bigg[d_1^n(v_1) \leq
2\sqrt{\log(\frac{n}{n-1})} \sqrt{n}\Bigg] > 1 - e^{-\log(\frac{n}{n-1})} =
\frac{1}{n}.
$$
\end{proof}
\subsection{Partial Results for Improved Degree Bounds in \\ Preferential Attachment Graphs }

According to Lemma~\ref{lem:urns_and_degrees}, a simple formula for the probability
distribution in the urn $[1,1,0,2]$ for arbitrary $a_0$ and $b_0$ leads
to a simple formula for the probability distribution of degrees in
preferential attachment graphs.
There is no closed formula for the $[1,1,0,2]$ matrix and
arbitrary $a_0$ and $b_0$ in the aforementioned work.
With Proposition~\ref{prop:general_case} by Flajolet, Dumas and Puyhaubert we
have an exact formula for the probability distribution of any balanced
triangular urn and arbitrary values $a_0$ and $b_0$.
\begin{proposition}[\cite{flajolet2006some}]
\label{prop:general_case}
For $(M,a_0,b_0)$ with $M =\begin{pmatrix}\alpha & \sigma-\alpha \\0 &
\sigma\end{pmatrix}$, we get
\begin{multline*}
    \P[A_n = a_0+k\alpha] = \\
    \frac{\Gamma(n+1)\Gamma(\frac{s_0}{\sigma})}{\Gamma(\frac{s_0}{\sigma} + n)}
    \binom{k+\frac{a_0}{\alpha}-1}{k}\sum_{i=0}^k (-1)^i
    \binom{k}{i}\binom{n+\frac{b_0-\alpha i}{\sigma}-1}{n}.
\end{multline*}
\end{proposition}
Unfortunately, this result contains a difficult alternating sum.
In the following lemma, we observe that
the alternating sum collapses
for the special case $\sigma=2$, $\alpha = 1$ and $b_0 = 0$.
This yields a simpler closed expression.
This result can be seen as a generalization of
Proposition~\ref{prop:easy_case} to arbitrary values for $a_0$.

\begin{lemma}
\label{lem:arbitrary_a0}
For $([1,1,0,2],a_0,0)$, we get
$$
    \P[A_n=a_0+k]=\frac{\Gamma(\frac{a_0}{2})\Gamma(\frac{1}{2}+n)}
    {\Gamma(\frac{a_0}{2}+n)\Gamma(\frac{1}{2})}
    \binom{k+a_0-1}{k}\frac{k}{n} 2^{k} \frac{{2n-k-1 \choose n-1}}
    {{2n \choose n}}.
$$
\end{lemma}
\begin{proof}
    The sum in Proposition~\ref{prop:general_case} is hard to simplify but by
    setting $M =[1,1,0,2]$ with $\sigma=2$, $a_0=1$ and $b_0=0$ and using the
    fact that the two values of Proposition~\ref{prop:easy_case} and
    Proposition~\ref{prop:general_case} have to be equal to $\P[A_n=1+k]$, we
    get
    $$
    \frac{\Gamma(n+1)\Gamma(\frac{1}{2})}{\Gamma(\frac{1}{2}+n)}
    \sum_{j=0}^k (-1)^j \binom{k}{j} \binom{n-\frac{j}{2} -1}{n}=
    \frac{k}{n} 2^{k} \frac{{2n-k-1 \choose n-1}}{{2n \choose n}}.
    $$
    Now we can get a closed formula for the sum:
    $$
    \sum_{i=0}^k (-1)^i
    \binom{k}{i}\binom{n-\frac{i}{2}-1}{n} = \frac{k}{n} 2^{k} \frac{{2n-k-1
    \choose n-1}}{{2n \choose n}}
    \frac{\Gamma(\frac{1}{2}+n)}{\Gamma(n+1)\Gamma(\frac{1}{2})}
    $$
    Since the sum does not depend on $a_0$ we directly get a closed expression
    for arbitrary $a_0$ and $b_0=0$ by inserting that value and $M$ into
    Proposition~\ref{prop:general_case}.
\end{proof}

For the rest of this section we want to find a simpler formula for the
probability distribution of the $[1,1,0,2]$ urn for arbitrary values $a_0$,
$b_0$. We want to avoid an alternating sum as in
Proposition~\ref{prop:general_case}.
For this, we first consider a simpler urn model with replacement matrix
$[2,0,0,2]$. This is a scaled version of the original \Purn with matrix
$[1,0,0,1]$.


\begin{proposition}[\cite{flajolet2006some}]\label{prop:2002}
Let $I=[2,0,0,2]$ and $a_0,b_0$ arbitrary,

$$
\P[A_n = a_0 + 2k] =
\frac{{a_0/2+k-1 \choose a_0/2-1}{n+b_0/2-k-1 \choose b_0/2-1}}
{{(a_0+b_0)/2+n-1 \choose (a_0+b_0)/2-1}}
$$
\end{proposition}

Now we can rewrite the probability distribution of the $[1,1,0,2]$ urn
for arbitrary $a_0,b_0$
as a sum involving the distribution on the
$[2,0,0,2]$ urn and the $[1,1,0,2]$ urn restricted to $b_0=0$.
For these two settings closed formulas are given by Lemma~\ref{lem:arbitrary_a0}
and Proposition~\ref{prop:2002}.
Since the sum is a sum of probabilities, every summand is non-negative.

\begin{lemma}
\label{lem:split_urns}
Let $M =[1,1,0,2]$ and $I=[2,0,0,2]$,
\begin{multline*}
    \P\big[A(M,n,a_0,b_0)= a_0 + k\big] \\= 
    \sum_{i=0}^n \P\big[A(M,i,a_0,0)=a_0+k\big] \P\big[A(I,n,a_0,b_0)= a_0 + 2i\big]
\end{multline*}
\end{lemma}
\begin{proof}
When we look at the urn $([1,1,0,2],a_0,b_0)$ after $n$ steps,
we distinguish between three
kinds of balls: Balls of type $A$, balls of type $B$ which were added after drawing
a ball of type $A$ and all other balls of type $B$.
Let us call the last two subtypes $B_1$ and $B_2$, respectively.
We define the type $A'$ to be the union of type $A$ and $B_1$.
Let $A'_n$ be the number of balls of this type
in the urn after $n$ steps.
The law of total probability states
\begin{multline*}
    \P\big[A(M,n,a_0,b_0)= a_0 + k\big]  \\=
    \sum_{i=0}^\infty \P\big[A(M,n,a_0,b_0)= a_0 + k \bigm|  A'_n=a_0+2i\big] \P\big[A'_n= a_0 + 2i\big].
\end{multline*}
Initially there are $a_0$ balls of type $A'$  and $b_0$ balls of type $B_2$.
If a ball of type $A'$ was drawn then two more balls of type $A'$ are added
and if a ball of type $B_2$ was drawn then two more balls of type $B_2$ are added.
Therefore
$$
\P\big[A'_n= a_0 + 2i\big] = \P\big[A(I,n,a_0,b_0)= a_0 + 2i\big].
$$
Observe that $\P\big[A(I,n,a_0,b_0)= a_0 + 2i\big] = 0$ for $i > n$.
Now we will consider $\P\big[A(M,n,a_0,b_0)= a_0 + k \bigm|  A'_n=a_0+2i\big]$,
i.e., the probability that $k$ balls of color $A$ were added under the condition
that $2i$ balls of color $A'$ were added.
Initially there are $a_0$ balls of color $A$ and zero balls of color $B_1$.
In each round the number of balls of color $A$ increases by one
if and only if the number of balls of color $A'$ increases by two.
There are two possibilities that the number of balls of color $A'$ may increase:
Either a ball of color $A$ was drawn
and a ball of color $A$ and $B_1$ was added or a ball of color $B_1$ was drawn
and two balls of color $B_1$ were added.
Therefore the distribution of balls of color $A$ and $B_1$ among color $A'$
can be modeled by
$$
    \P\big[A(M,n,a_0,b_0)= a_0 + k \bigm|  A'_n=a_0+2i\big] = \P[A(M,i,a_0,0) = a_0 + k].
$$
\end{proof}

If we replace the probabilities in the sum of the previous lemma
with the closed formulas from Lemma~\ref{lem:arbitrary_a0}
and Proposition~\ref{prop:2002} we get the following formula
with positive summands.

\begin{lemma}
\label{lem:non_alternating}
Given
$([1,1,0,2],a_0,b_0)$,
$$
    \P[A_n=a_0+k]=
    \frac{\Gamma(\frac{a_0}{2})\binom{k+a_0-1}{k}2^{k}k}{\Gamma(\frac12){(a_0+b_0)/2+n-1
            \choose (a_0+b_0)/2-1}} \sum_{i=0}^n \frac{\Gamma(\frac12+i){2i-k-1
    \choose i-1}{a_0/2+i-1 \choose a_0/2-1}{n+b_0/2-i-1
\choose b_0/2-1}}{\Gamma(\frac{a_0}{2}+i)i{2i \choose i}}
$$
\end{lemma}
\begin{proof}
Combining Lemma~\ref{lem:arbitrary_a0}, Proposition~\ref{prop:2002}
and Lemma~\ref{lem:split_urns}.
\end{proof}

\begin{theorem}
\label{thm:non_alternating}
$$
    \P\big[\D[n][1][v_t]=k\big]=    \frac{2^{k}k}{{n-1/2
    \choose t-3/2}} \sum_{i=0}^{n-t+1} \frac{{2i-k-1
    \choose i-1}{i-1/2 \choose -1/2}{n-1-i
    \choose t-2}}{i{2i \choose i}}
$$
\bigskip
\begin{multline*}
\P\big[\D[n]=k \bigm| \D[t] =d\big]\\ =
\frac{\Gamma(\frac{d}{2})\binom{k+d-1}{k}2^{k}k}{\Gamma(\frac{1}{2}){n-3/2
    \choose t-3/2}} \sum_{i=0}^{n-t} \frac{\Gamma(\frac{1}{2}+i){2i-k-1
    \choose i-1}{d/2+i-1 \choose
    d/2-1}{n-d/2-i-3/2
    \choose t - 3/2 - d/2}}{\Gamma(\frac{d}{2}+i)i{2i \choose i}}
\end{multline*}
\end{theorem}
\begin{proof}
Inserting the values of Lemma~\ref{lem:urns_and_degrees} into
Lemma~\ref{lem:non_alternating}.
\end{proof}

By employing algebraic techniques we were also able to transform
the formula in Lemma~\ref{lem:non_alternating} into a generalized
hypergeometric function, for which further methods
exist~\cite{graham1994concrete}.

\begin{proposition}\label{prop:hypergeom}
Given $([1,1,0,2],a_0,b_0)$,
$$
    \P[A_n=a_0+k]  =- \frac{\Gamma(\frac{a_0}{2})\binom{k+a_0-1}{k} k 2^{2k-1}
    \Gamma \left(\frac{a_0+b_0}{2}\right) \Gamma \left(\frac{b_0}{2}+n\right)
    \left(g(n,k,a_0,b_0)-1\right)}{\Gamma \left(\frac{b_0}{2}\right) \Gamma
    \left(\frac{a_0+b_0}{2}+n\right) \Gamma(\frac{1}{2}){(a_0+b_0)/2+n-1
    \choose (a_0+b_0)/2-1}}
$$
where $g(n,k,a_0,b_0)= F(\frac a2, \frac12 - \frac k2, -\frac k2,-n;1,1-k,-
\frac b2-n+1;1)$ and $F$ is the generalized hypergeometric function.
\end{proposition}

It remains an open question whether the formulas in
Lemma~\ref{lem:non_alternating} or Proposition~\ref{prop:hypergeom}
can be further simplified or approximated to improve
the bounds from Section~\ref{sec:bounds}.

\subsection{Remark: \Purns and Graphs with More Edges}
The previous subsections gave results for the preferential attachment
process with parameter $m=1$.
One can use the reduction in Section~\ref{sec:prelim:model}
generalize these results to arbitrary $m$.
However, Farczadi and Wormald~\cite{farczadi2014degree} 
defined $G_m$ directly via a balanced urn model and wee briefly state their
formulation for the distribution of $d_n(v_i)$.
They set
\begin{align*}
a_0 &= m\\
b_0 &= 2i\\
s_0 &= a_0 + b_0 = 2i,
\end{align*}
use the replacement matrix
$$
M = \begin{pmatrix}m & m \\0 & 2m\end{pmatrix},
$$
and do $n' = n-i$ trials. Using Flajolet~et~al.'s methods they obtain
\begin{align*}
\E[d_n(v_i)] &=
\frac{m\Gamma(i)\Gamma(n+\frac{1}{2})}{\Gamma(i+\frac{1}{2})\Gamma(n)}\\
\E[d_n(v_i)^2] &= \frac{m(m+1)n}{i} + 2 \E[d_n(v_i)].
\end{align*}
They show how urn models can help (re-)proving statements about
preferential attachment graphs.


\section{Preferential Attachment Graphs are \\ Somewhere-Dense}
\label{sec:somewhere-dense}

\def\A{A}
\def\B{B}
\def\Aa{\cal A}
\def\Bb{\cal B}

In this section we show that preferential attachment graphs are \aas somewhere-dense.
We do so by analyzing the probability that a preferential attachment graph of size $n$
contains a one-subdivided clique of size $k:=\log(n)^{1/4}$ as a subgraph.
Let this probability be $p_n$. We show that $\lim_{n\to\infty}p_n=1$.
The proof works as follows:
We start with a small preferential attachment graph and pick a set of $k$
vertices with high degree.
These are the principal vertices.
We then add vertices to the graph according to the preferential attachment process. A
one-subdivided clique of size $k$ arises, if for every pair of principal vertices $v$ and $w$,
a new vertex $u$ is added that is adjacent to both $v$ and $w$.
The core of this proof is to show that after $n=2^{k^4}$
vertices have been inserted, with high probability
there is at least one connecting vertex for every pair of principal vertices.

We now describe an urn experiment
that illustrates for a pair of principal vertices $v$,$w$ the probability
that a new vertex $u$ is connected to both $u$ and $v$.
This experiment has no connection to \Purns and is solely used for illustrative
purposes.
The experiment consist of multiple rounds.

In the $i$th round (for simplicity we assume $i \ge 10$),
we define an urn containing $i$ balls, where
$\sqrt{\lceil i \rceil}$ balls are red, $\sqrt{\lceil i \rceil}$ balls are blue,
and the rest is black.
In each round we draw two balls uniformly at random from the urn.
The experiment succeeds if we draw a red and a blue ball in the same round.
It is easy to see that the probability of success in the $i$th round
equals $2(\sqrt{\lceil i \rceil}/i)^2$.
We observe that eventually the experiment succeeds because
$$
1 - \prod_{i=10}^{\infty} \Bigl( 1 - 2 \bigl(\frac{\sqrt{i}}{i}\bigr)^2 \Bigr) = 1.
$$
This experiment behaves similarly to the process of connecting two principal
vertices. Two principal vertices are connected in round $i$ if the vertex $v_i$
is connected to both principal vertices.
The expected degree of the first vertex in a preferential attachment graph of size $i$
is proportional to $\sqrt{i}$.
If we (naively) assume that the degrees of $u$ and $v$
at time $i$ are at all times exactly $\sqrt{i}$
then a new vertex throws an edge to $v$ or $w$ with probability
proportional to $\sqrt{i}/i$ and is connected to both with probability
roughly $2(\sqrt{i}/i)^2$.
Therefore, the probability that in the $i$th step a new vertex
$u$ connects both $v$ and $w$
is proportional to the probability that in the $i$th round of the urn
experiment a red and a blue ball is drawn.
Using similar arguments we show in this section
that the success probability of building a one-subdivided clique
also is high.

If we, however, alter the urn experiment and assume that in the $i$th round
there are only $\sqrt{i}/\log(i)$ red or blue balls we cannot guarantee
success because
$$
1 - \prod_{i=10}^{\infty}
\Bigl( 1 - 2 \bigl(\frac{\sqrt{i}/\log(i)}{i}\bigr)^2 \Bigr) \neq 1.
$$
This means if the expected value of the degrees were just a logarithmic
factor smaller then our proof would not work.
This suggests that preferential attachment graphs are ``just barely''
\aas somewhere-dense\footnote{
    In a follow up work we will show that after removing a small number of vertices
    from a preferential attachment graph, the probability is very high that
    the graph belongs to a nowhere dense graph class.
}.
This also means we need lower bounds which guarantee that
the degrees of our principal vertices are
not much smaller than their expected value, e.g., at most
a constant factor off.
Bounds which guarantee a factor of $1/\log(i)$ would not be sufficient.

Unfortunately the probability distribution
of the degree of a vertex is only well centered around its expected value
if its initial degree is already large (see right side of Figure~\ref{fig:vertex_conditional}).
We use Lemma~\ref{lem:we-got-principal-vertices} to find in a graph of size $k^4$
with high probability $k$ principal vertices which have a degree of roughly $k$.
These vertices will be our principal vertices.
Their degree is centered closely around their expected value
and therefore close enough to $\sqrt{i}$ for our proof to work.

Our main proof is done in Theorem~\ref{thm:degree-bounds}.
There we argue that with high probability
for every pair of principal vertices there will eventually be a vertex
which is connected to both.
At first we tried to prove this by showing that with high
probability the degrees of the principal vertices are always high
(roughly $\sqrt{i}$ in the $i$th round for every $i$)
and then showing that the probability that the principal vertices will be connected
under the condition that the degrees of the principal vertices are always high.
This method did not work out because of the subtle dependencies
between the event that the $i$th vertex connects two principal vertices
and the event that the principal vertices have a certain degree in round
$j \ge i$.

One can, however, easily compute the probability that a new edge of the $i$th vertex
is connected to a principal vertex under the (weaker) condition that the principal vertex
has at least a certain degree at time $i-1$.
Let $\B_i$ be the event that the degree of a vertex is at least half the expected
degree at time $i$. The event $\B_i$ on its own is very likely and gives us a
good bound for the probability that $v_{i+1}$ is connected to the principal vertex.
Our calculations work in a way where whenever we assume a new event $\B_i$ the
probability $\P(\bar \B_i)$ is added to our failure probability
(see Lemma~\ref{lem:squence_A_B}). So if we were
to assume all events $\B_i,\B_{i+1},\B_{i+2},\dots$ our bound quickly becomes
meaningless as the sum $\P(\bar \B_i)+\P(\bar \B_{i+1})+\P(\bar \B_{i+2})+\dots$
becomes larger than one.
But if we assume exponentially spaced events
$\B_i,\B_{2i},\B_{4i},\B_{8i},\dots$ our bound on the failure probability stays
small enough and new vertices are still likely to be connected to our principal vertices,
allowing us to show in Theorem~\ref{thm:somewhere_dense}
that $G_2^n$ contains a large clique.

We now proceed to prove the results of this section.
The first Lemma shows that there are some
vertices which have a reasonably high degree after
a short number of steps.
\begin{lemma}\label{lem:we-got-principal-vertices}
    Let $64000\le k \in \bf{N}$. With a probability of at least
    $1 - 2ke^{-c k}$ (for some positive constant factor $c$) there exists a set
    of vertices
    $X \subseteq \{ v_1,\dots,v_{k^2} \}$, $|X|=k$
    such that $\D[k^{4}][m][x] \ge \frac12 mk$ for all $x \in X$.
\end{lemma}
\begin{proof}
We partition the first $k^{2}$ vertices into $k$ sets of $k$ vertices. Let $S$
be one of these sets. We know that $\D[k^2][m] \ge mk$.
Theorem~\ref{thm:degree-bounds} therefore implies with $t=k^2$, $n=k^4$ and
$\varepsilon = 1/40$
$$
\P\Bigl[\D[k^{4}][m] \le \frac12 m k^2 \Bigr]
\le \P\Bigl[\D[k^{4}][m] \le (1-0.1) \sqrt{\frac{k^{4}}{k^2}} mk \Bigr]
\le 2e^{-c k}
$$
where $c$ is the positive constant factor of Theorem~\ref{thm:degree-bounds}.
This theorem additionally requires $k^2 = t \ge 4/\varepsilon^3$
which is satisfied for $k \ge 64$.
With a probability of at least
$1 - k 2e^{-c k}$ each of the $k$ sets have at time $k^4$ a total degree of at
least $\frac12 mk^2$ by the union bound.
Let now $x_i$ be the vertex in the $i$th set that has the highest
degree after $k^{4}$ steps and let $X = \{x_1,\dots,x_k\}$.
Since every set contains at most $k$ vertices, the
vertex with the highest degree has a degree of at least $\frac12 m k$.
\end{proof}

We now bound the probability that two principal vertices $v_a,v_b$ become
connected under the condition that they have high degree.

\begin{lemma} \label{lem:sequence_\A_H}
We consider the preferential attachment process with $m \ge 2$.
Let $k \in \bf N$ and $v_a,v_b$ be any vertices.
Let $\B_i$ be the event that $\D[i][m][v_a], \D[i][m][v_b] \ge m\sqrt{i}/4k$.
Let $\A_{j,i}$ with $j \le i$ be the event
that the first two edges of at least one of the vertices $v_j,\dots,v_i$ 
are adjacent to $v_a$ and $v_b$, respectively.
Then $\P(\bar \A_{i+1,2i} \mid \bar \A_{j,i}, \B_{i}) \le e^{-\frac{1}{256k^2}}$
for $k^4 \le i$ and $j \le i$.

\end{lemma}
\begin{proof}
Let $u > 0$.
$\P(\A_{i+u,i+u} \mid \B_i)$ is the
probability that vertex $v_{i+u}$ is adjacent to both $v_a$ and $v_b$
under the condition that $v_a$ and $v_b$ have degree
at least $m\sqrt{i}/4k$ at some earlier time $i$.
When vertex $v_{i+u}$ is inserted, the random process draws $m \ge 2$ edges
from $v_{i+u}$ to earlier vertices.
The probability that some vertex is chosen equals its degree divided
by the total number of edges in the graph at this time.
The degree of $v_a$ and $v_b$ is at
least $m\sqrt{i}/4k$ at this point in time.
Also there is a total of at most $2(i+u)m$ edges in the graph.
We can therefore bound
$$
\P(\A_{i+u,i+u} \mid \B_i) \ge \Bigl( \frac{\sqrt{i}}{8(i+u)k} \Bigr)^2.
$$
The same argument holds if we additionally assume some
of the earlier vertices not to be adjacent to both $v_a$ and $v_b$.
Let $j < i$. Then
$$
\P(\A_{i+u,i+u} \mid \bar \A_{j,i+u-1}, \B_i) \ge \Bigl( \frac{\sqrt{i}}{8(i+u)k} \Bigr)^2.
$$
We now consider the probability that no vertex
in a sequence of vertices is adjacent to both $v_a$ and $v_b$.
The chain rule yields
\begin{align}
    \P(\bar \A_{i+1,2i} \mid \bar \A_{j,i}, \B_i) \nonumber
    & = \prod_{u=1}^{i} \P(\bar \A_{i+u,i+u} \mid \bar \A_{j,i+u-1}, \B_i) \nonumber \\
    & \le \prod_{u=1}^{i} \Bigl( 1 - \Bigl( \frac{\sqrt{i}}{8(i+u)k} \Bigr)^2\Bigr) \nonumber \\
    & \le \Bigl( 1 - \Bigl( \frac{\sqrt{i}}{16ik} \Bigr)^2 \Bigr)^i \nonumber 
     \le \Bigl( 1 - \frac{1}{256ik^2} \Bigr)^i \nonumber 
     \le e^{-\frac{1}{256k^2}}. \nonumber
\end{align}
\end{proof}

Imagine a sequence of events $\Aa_0,\dots,\Aa_{l-1}$
such that a preferential attachment graph contains
a large subdivided clique if
any one of these events is false.
This means it is sufficient to show that the probability $\P[\Aa_0\cap\dots\cap\Aa_{l-1}]$
is small. Assume we can only bound the probability of event $\Aa_i$ under the 
condition $\Bb_i$.
The following lemma gives a good approximation of $\P[\Aa_0\cap\dots\cap\Aa_{l-1}]$
if the events $\Bb_i$ have a high probability.
\begin{lemma}
\label{lem:squence_A_B}
Let $\Aa_0,\dots,\Aa_{l-1},\Bb_0,\dots,\Bb_{l-1}$ be events.
Then
$$
    \P[\bar \Aa_0\cap\dots\cap\Aa_{l-1}] \le \sum_{0=1}^{l-1} \P[\bar \Bb_i] +
    \prod_{0=1}^{l-1} \P[\bar \Aa_i \mid \bar \Aa_0\cap\dots\cap\bar \Aa_{i-1}\cap\Bb_i].
$$
\end{lemma}
\begin{proof}
As a first step we apply the chain rule and the law of total probability.
Let $i<l$. Then
\begin{align*}
    \P[\bar \Aa_0\cap\dots\cap\bar \Aa_{i}] &  =
        \P[\bar \Aa_0\cap\dots\cap\bar \Aa_{i-1}] \P[\bar \Aa_i \mid
            \bar \Aa_0\cap\dots\cap\bar \Aa_{i-1}] \\ &= \P[\bar \Aa_0\cap\dots\cap\bar \Aa_{i-1}]
            \P[\bar \Aa_i \mid \bar \Aa_0\cap\dots\cap\bar \Aa_{i-1}\cap\Bb_i] \P[\Bb_i]\\ &+
            \P[\bar \Aa_i \mid \bar \Aa_0\cap\dots\cap\bar \Aa_{i-1}\cap\bar \Bb_i] \P[\bar
            \Bb_i]\\ &\leq \P[\bar \Aa_0\cap\dots\cap\bar \Aa_{i-1}]
            \P[\bar \Aa_i \mid \bar \Aa_0\cap\dots\cap\bar \Aa_{i-1}\cap\Bb_i] + \P[\bar \Bb_i].
\end{align*}
We can now recursively apply this inequality, and use an upper bound of $1$ for
all factors in front of $\P[\Bb_j]$ when expanding the product, to get
\begin{align*}
    \P[\bar \Aa_0\cap\dots\cap\bar \Aa_{l-1}]&\leq \sum_{i=1}^{l-1} \P[\bar \Bb_i] +
    \prod_{i=1}^{l-1} \P[\bar \Aa_i \mid \bar \Aa_0\cap\dots\cap \bar \Aa_{i-1}\cap\Bb_i],
\end{align*}
which proves the claim.
\end{proof}

We now use Lemma~\ref{lem:we-got-principal-vertices}
and Lemma~\ref{lem:sequence_\A_H} to prove
the main result of this section.

\begin{theorem}
\label{thm:somewhere_dense}
    Let $m \ge 2$.
    $G_m^n$ contains
    with a probability of at least $1 - e^{-c \log(n)^{1/4}}$
    a one-subdivided clique of size
    $\lfloor \log(n)^{1/4} \rfloor$
    for some positive constant $c$.
\end{theorem}
\begin{proof}
Let $k \in \N$. We will prove this theorem
by showing that $k$ vertices in $G_m^n$ are with high probability pairwise connected by
a path of length two and thereby span a one-subdivided clique.
Later, the value for $n$ will be chosen as
$k^42^{k^{3}}$ which implies $k \ge \log(n)^{\frac{1}{4}}$ for $k \ge 2$.

We know for vertices with high degree that their degree will be centered
closely around their expected value in the future. 
Let us therefore assume that there are vertices $v_1,\dots,v_k$
such that $\D[k^4][m][v_i] \ge \frac12 m k$ for $1 \le i \le k$.
We call these vertices principal vertices.
Lemma~\ref{lem:we-got-principal-vertices} states that these principal vertices
exist with a probability of at least $1 - 2ke^{-ck}$ for some $c > 0$.
Let us fix a pair of principal vertices $v_a$, $v_b$ and show that with high probability
there is a vertex that is adjacent to both of these principal vertices.
The higher the degree of $v_a$ and $v_b$ the higher the probability that
a new vertex is adjacent to both $v_a$ and $v_b$.
For $i \ge k^4$ we define $\B_i$ to be the event
that $\D[i][m][v_a], \D[i][m][v_b] \ge m\sqrt{i}/4k$.
Since we assume $\D[k^4][m][v_a] \ge \frac12 m k$
Theorem~\ref{thm:degree-bounds} states with $t=k^4$ and $\varepsilon=\frac12$, that
$$
\P[\bar \B_i] =
\P\Bigl[\D[i][m][v_a] < \frac{m\sqrt{i}}{4k} = \frac12 \sqrt{\frac{i}{k^4}} \frac12 m k \Bigr]
\le 2e^{-c k}.
$$

We define $\A_{j,i}$ with $j \le i$ to be the event
that at least one of the vertices $v_j,\dots,v_i$ is adjacent to both $v_a$ and $v_b$.
We will prove the claim by showing that
$\P(\bar \A_{k^4+1,k^42^l})$ converges to zero for an appropriate choice of $l$.
We divide our vertices into windows which double in size.
We set $\bar \Aa_i = \bar \A_{k^42^i+1,k^42^{i+1}}$ to be the event
that none of the vertices $v_{i'+1},\dots, v_{2i'}$ is adjacent to $v_a$ and $v_b$
for $i'=k^42^i$.
It holds that $ \bar \A_{k^4+1,k^42^l} = \Aa_0\cap\dots\cap\Aa_{l-1} $.
Furthermore, by setting $\Bb_i = \B_{k^42^i}$, Lemma~\ref{lem:sequence_\A_H} states that
$\P(\bar \Aa_i \mid \bar \Aa_0\cap\dots\cap\bar \Aa_{i-1}\cap \Bb_i) \le e^{-\frac{1}{256k^2}}$.
By Lemma~\ref{lem:squence_A_B}
\begin{align*}
         \P(\bar \A_{k^4+1,k^42^l})
     =&   \P(\bar \Aa_0\cap\dots\cap\bar \Aa_{l-1}) \\
     \le& \sum_{i=0}^{l-1} \P(\bar \Bb_i) + \prod_{i=0}^{l-1} 
        \P(\bar \Aa_i \mid \bar \Aa_0\cap\dots\cap\bar \Aa_{i-1}\cap\Bb_i) \\
     \le& \sum_{i=0}^{l-1} 2e^{-ck} + e^{-\frac{1}{256k^2}}
     = 2lke^{-ck} + e^{-\frac{l}{256k^2}}.
\end{align*}
We define $l=k^{3}$, $n = k^42^l$ and
$$
p = k^3 2ke^{-ck} + e^{-\frac{k^3}{256k^2}} \ge \P(\bar \A_{k^4+1,n}).
$$
This means that in $G_m^n$ the probability that there exists a vertex
which connects the principal vertices $v_a$ and $v_b$ is at least $1-p$.
According to the union bound, the probability that for all ${k \choose 2}$
pairs of principal vertices there exists a vertex which connects them is
bounded by $1 - {k \choose 2} p$.
In Lemma~\ref{lem:sequence_\A_H}, only
the first two edges of the connecting vertex are considered.
Therefore a connecting vertex may only connect a single pair of principal vertices.
This means every pair of principal vertices has a unique connecting vertex, i.e.,
the principal vertices span a one-subdivided clique.

So far, all our calculations were based on the condition
that there are $k$ principal vertices with reasonably high degree in the
beginning.
According to Lemma~\ref{lem:we-got-principal-vertices}, the probability $k$
such vertices do not exist is at most $2ke^{-ck}$ for some $c > 0$.
So by law of total probability, we can add this error
probability to the conditional bound to get an unconditional bound.
This means that $G_m^n$ contains no one-subdivided $k$
clique with a probability of at most
$$
    {k \choose 2} p + 2ke^{-ck} \le \text{poly}(k)e^{-c'k}
    = \text{polylog}(n) e^{-c' \log(n)^{\frac14}}
    \le e^{-c'' \log(n)^{\frac14}}
$$
for some positive constant $c''$.
\end{proof}

The main result of this section follows directly
from the previous theorem and Definition~\ref{def:aas-somewhere-dense}.
\begin{corollary}
$G_m^n$ is \aas somewhere-dense for $m\geq2$.
\end{corollary}


\section{Conclusion}
\label{sec:conclusion}

Our analysis of preferential attachment graphs resulted in two main results:
(1) a tail bound result stating the degree of vertices, under the condition
that the degree had a certain value $D(t)$ at an earlier point in time
$t$, is a small factor away from its expectation is exponentially small in
$D(t)$ and (2) that preferential attachment graphs are \aas somewhere-dense.
For the first result we used Chernoff bounds to obtain bounds which are
reasonably tight for small time-frames and then used the union bound to extend
these bounds to large time-frames.
The second result was
obtained by using the fact that even though for a single vertex the
probability distribution of its degree is not concentrated
the probability distribution of sets of vertices is indeed concentrated.
Among such a set we choose the principal vertices for a large clique-minor.

Recently a more general preferential attachment model has been introduced that
has an additional parameter $\delta$~\cite{hofstad1}. This new parameter
captures how much the degree of a vertex changes the probability that another
vertex connects to it, where $\delta = 0$ is the standard model considered in
this paper and $\delta = \infty$ corresponds to uniform attachment where the
degrees do not matter. Another notable case is $\delta \leq -1$, where
self-loops no longer occur.
It would be interesting to see for what values of
$\delta$ this model is \aas somewhere- or nowhere-dense and if our
techniques can still be applied.

The way we (sequentially) constructed the graph follows the
definition by \Bollobas and Spencer~\cite{Bollobas:2001} but since the
original definition by
\Barabasi and Albert~\cite{barabasi1999emergence} was ambiguous there are two
other natural ways to define the preferential attachment
step~\cite{berger2014asymptotic}: In the \emph{Independent Model} a new vertex
$v_t$ chooses its $m$ neighbors $w_1,\ldots,w_m$ independently at the same
time with repetitions and the \emph{Conditional Model} is the same but the
$w_1,\ldots,w_m$ have to be different. The model we use is known as the
\emph{Sequential Model}, where the neighbors are chosen after one another
(which means the probabilities are updated in-between). It has the nice
property that it can also be easily modeled as a \Purn, which is not the case
for the other interpretations. It would be
interesting to generalize our results in Section~\ref{sec:bounds} and
\ref{sec:somewhere-dense} to the Independent and the Conditional Model.

In Section~\ref{sec:urns} we tried to use \Purns to find simple formulas
for the probability distribution of vertex degrees.
It is open if this approach can be used to improve the bounds presented in
Section~\ref{sec:bounds}. 
On the other hand it might be interesting to see if
the bounds from Section~\ref{sec:bounds} can be used to improve certain
bounds for \Purns. 
Nice bounds for the original \Purn with replacement matrix $[1,0,0,1]$
were obtained using martingale techniques.
Since martingales were not applicable our setting we
used Chernoff bounds on multiple intervals.
This technique might be flexible enough to provide bounds
for other replacement matrices of \Purns as well.

\bibliographystyle{abbrv}
\bibliography{references,conferences}

\end{document}

%% file: header.tex
\usepackage{amssymb,amsmath,amsthm,amsfonts}
\usepackage[colorinlistoftodos]{todonotes}
\usepackage{xspace,xpunctuate}
\usepackage{xfrac}
\usepackage{thmtools}
\usepackage{hyperref}
\usepackage{mathtools}
\usepackage{enumerate}
\usepackage{subcaption}
\usepackage[a4paper]{geometry}
\usepackage[english]{babel}
\usepackage{ifpdf}
\usepackage{graphicx}
\ifpdf
\fi

\newtheorem{definition}{Definition}[section]
\newtheorem{proposition}[definition]{Proposition}

\newtheorem{lemma}[definition]{Lemma}
\newtheorem{theorem}[definition]{Theorem}
\newtheorem{corollary}[definition]{Corollary}

\newtheorem{innercustomthm}{Theorem}
\newenvironment{customthm}[1]
  {\renewcommand\theinnercustomthm{#1}\innercustomthm}
  {\endinnercustomthm}

\def\Nesetril{Ne\v{s}et\v{r}il\xspace}

\def\Barabasi{Barab\'{a}si\xspace}
\def\Bollobas{Bollob\'{a}s\xspace}
\def\BA{\Barabasi--Albert\xspace}
\def\BaAlGrs{\BA graphs\xspace}
\def\Polya{P\'{o}lya}
\def\Purn{\Polya~Urn\xspace}
\def\Purns{{\Purn}s\xspace}

\def\aas{a.a.s\xperiod}

\newcommand{\N}{\ensuremath{\mathbf{N}}\xspace} 
\newcommand{\R}{\ensuremath{\mathbf{R}}\xspace} 
\renewcommand{\cal}{\mathcal}

\newcommand{\topnab}{\mathop{\widetilde \triangledown}}
\def\topgrad_#1{\widetilde \nabla\!_{#1}}

\renewcommand{\P}{\operatorname{P}}
\newcommand{\E}{\operatorname{E}}